%% file: main.tex
\documentclass[11pt]{article}
\usepackage{fullpage, graphicx}
\usepackage{amsfonts, amsmath, amsthm, amssymb, mathtools,enumitem}
\usepackage{algorithm}
\usepackage{nicefrac}
\usepackage[noend]{algpseudocode}

\include{marginnotes}
\include{macros}

 \newcommand{\calU}{\mathcal{U}}
 \newcommand{\calV}{\mathcal{V}}
 \newcommand{\calW}{\mathcal{W}}
 \newcommand{\abs}[1]{\lvert #1 \rvert }

\title{Smoothed Analysis in Unsupervised Learning via Decoupling}
\author{Aditya Bhaskara\thanks{School of Computing, University of Utah. Email: \textsf{bhaskaraaditya@gmail.com}.} \and Aidao Chen\thanks{Northwestern University. Email: \textsf{aidao@u.northwestern.edu}.} \and Aidan Perreault\thanks{Northwestern University. Email: \textsf{aperreault@u.northwestern.edu}.} \and Aravindan Vijayaraghavan\thanks{Northwestern University. Email: \textsf{aravindv@northwestern.edu}. The second, third and last authors were supported by the National Science Foundation (NSF) under Grant No.~CCF-1652491 and CCF-1637585. Additionally, the third author was supported by an undergraduate research grant from Northwestern University. }}

\date{}
\begin{document}
\maketitle

\begin{abstract}

Smoothed analysis is a powerful paradigm in overcoming worst-case intractability in unsupervised learning and high-dimensional data analysis. While polynomial time smoothed analysis guarantees have been obtained for worst-case intractable problems like tensor decompositions and learning mixtures of Gaussians, such guarantees have been hard to obtain for several other important problems in unsupervised learning. A core technical challenge in analyzing algorithms is obtaining lower bounds on the least singular value for random matrix ensembles with dependent entries, that are given by low-degree polynomials of a few base underlying random variables.

In this work, we address this challenge by obtaining high-confidence lower bounds on the least singular value of new classes of structured random matrix ensembles of the above kind. 
We then use these bounds to design algorithms with polynomial time smoothed analysis guarantees for the following three important problems in unsupervised learning:

\begin{itemize}
\item {\em Robust subspace recovery}, when the fraction $\alpha$ of inliers in the $d$-dimensional subspace $T \subset \R^n$ is at least $\alpha > (d/n)^\ell$ for any constant $\ell \in \Z_+$. This contrasts with the known worst-case intractability when $\alpha< d/n$, and the previous smoothed analysis result which needed $\alpha > d/n$ (Hardt and Moitra, 2013).   
\item {\em Learning overcomplete hidden markov models}, where the size of the state space is any polynomial in the dimension of the observations. This gives the first polynomial time guarantees for learning overcomplete HMMs in the smoothed analysis model.
\item {\em Higher order tensor decompositions}, where we generalize and analyze the so-called FOOBI algorithm of Cardoso to find order-$\ell$ rank-one tensors in a subspace. This allows us to obtain polynomially robust decomposition algorithms for $2\ell$'th order tensors with rank $O(n^{\ell})$. 
  
\end{itemize}   
\anote{4/4: changed order of HMM and FOOBI}
\end{abstract}

\section{Introduction}
\input{intro.tex}

\section{Preliminaries} \label{sec:prelims}
\anote{Need to write down various notations etc.}
\cnote{I use the fact that the rank of $x^\ell$ is at most $\binom{n+\ell-1}{\ell}$, is it trivial for reviewer, or should we mention it here?} \anote{I mentioned it towards the end of prelims.}
In this section, we introduce notation and preliminary results that will be used throughout the rest of the paper.

Given a vector $a \in \R^n$ and a $\rho$ (typically a small inverse polynomial in $n$), a {\em $\rho$-perturbation} of $a$ is obtained by adding independent Gaussian random variables $x_i \sim N(0, \rho^2/n)$ to each coordinate of $a$. The result of this perturbation is denoted by $\ta$.

We will denote the singular values of a matrix $M$ by $\sigma_1(M), \sigma_2(M), \ldots$, in decreasing order. We will usually use $k$ or $R$ to represent the number of columns of the matrix.  The maximum and minimum (nonzero) singular values are also sometimes written $\sigma_{max}(M)$ and $\sigma_{min}(M)$.

While estimating the minimum singular value of a matrix can be difficult to do directly, it is closely related to the {\em leave-one-out distance} of a matrix, which is often much easier to calculate.
\begin{definition}\label{def:leave-one-out}
Given a matrix $M \in \R^{n \times k}$ with columns $M_1, \ldots, M_k$, the leave-one-out distance of $M$ is
\begin{equation}
    \ell(M) = \min_i dist(M_i, \mathrm{Span}\{M_j : j \neq i\}).
\end{equation}
\end{definition}
The leave-one-out distance is closely related to the minimum singular value, up to a factor polynomial in the number of columns of $M$~\cite{RudelsonV}.
\begin{lemma}\label{lem:leave-one-out}
For any matrix $M \in \R^{n\times k}$, we have 
\begin{equation}
    \frac{\ell(M)}{\sqrt{k}} \le \sigma_{min}(M) \le \ell(M).
\end{equation}
\end{lemma}

\paragraph{Tensors and multivariate polynomials.}
An order-$\ell$ tensor $T \in \R^{n \times n \times \dots \times n}$ has $\ell$ modes each of dimension $n$. Given vectors $u, v \in \R^n$ we will denote by $u \otimes v \in \R^{n \times n}$ the outer product between the vectors $u, v$, and by $u^{\otimes \ell}$ the outer product of $u$ with itself $\ell$ times i.e., $u \otimes u \otimes \dots \otimes u$. 

We will often identify an $\ell$th order tensor $T$ (with dimension $n$ in each mode) with the vector in $\R^{n^{\ell}}$ obtained by flattening the tensor into a vector. For sake of convenience, we will sometimes abuse notation (when the context is clear) and use $T$ to represent both the tensor and flattened vector interchangeably. Given two $\ell$th order tensors $T_1, T_2$ the inner product $\iprod{T_1, T_2}$ denotes the inner product of the corresponding flattened vectors in $\R^{n^{\ell}}$. 

A symmetric tensor $T$ of order $\ell$ satisfies $T(i_1, i_2, \dots, i_\ell) = T(i_{\pi(1)}, \dots, i_{\pi(\ell)})$ for any $i_1, \dots, i_\ell \in [n]$ and any permutation $\pi$ of the elements in $[\ell]$. It is easy to see that the set of symmetric tensors is a linear subspace of $\R^{n^{\otimes \ell}}$, and has a dimension equal to ${n+\ell-1 \choose \ell}$. 
Given any $n$-variate degree $\ell$ homogenous polynomial $g \in \R^{n} \to \R$, we can associate with $g$ the unique symmetric tensor $T$ of order $\ell$ such that $g(x)=\iprod{T, x^{\otimes \ell}}$.

\anote{4/2: added Anari et al.}
\paragraph{Minimum singular value lower bounds for decoupled tensor products.} We will use as a black box high confidence lower bounds on the minimum singular value bounds for decoupled tensor products. The first statement of this form was shown in \cite{BCMV}, but this had a worse polynomial dependence on $n$ in both the condition number and the exponent in the failure probability. The following result in \cite{ADMPSV18} gives a more elegant proof, while also achieving much better bounds in both the failure probability and the minimum singular value.  
\begin{lemma}[\cite{ADMPSV18}, Lemma 6]\label{lem:quantitativebound}
Let $p \in (0,1], \delta\in(0,1)$ be constants, and let $W\subseteq \mathbb{R}^{n^{\otimes\ell}}$ be an arbitrary subspace of dimension at least $\delta n^\ell$. Given any $x_1,\cdots, x_\ell\in\mathbb{R}^n$, then for their random perturbations $\tilde{x}_1,\cdots,\tilde{x}_\ell$ where for each $i\in[\ell]$, $\tilde{x}_i=x_i+N(0,\rho_i^2/(2n\ell))^n$ with $\rho_i^2\geq \rho^2$, we have
\begin{align*}
    \mathbb{P}\Big[\norm{\Pi_W (\tilde{x}_1\otimes\tilde{x}_2\otimes\cdots\otimes\tilde{x}_\ell)}_2< \frac{c_1(\ell)\rho^\ell}{n^\ell}\cdot p^\ell\Big]\leq p^{c_2(\ell)\delta n}
\end{align*}
where $c_1(\ell), c_2(\ell)$ are constants that depend only on $\ell$.
\end{lemma}
We remark that the statement of Anari et al.~\cite{ADMPSV18} is stated in terms of the distance to the orthogonal subspace $W^{\perp}$, as long as $dim(W^{\perp}) \le c^\ell n^\ell$ for some $c<1$; this holds above for $c=1-\delta/ \ell$.\cnote{4.22: changed to $1-\delta/\ell$} 


\anote{Davis-Kahan theorem?}

\input{decoupling.tex}

\section{Polynomials of Few Random Vectors}
\input{polynomials.tex}

\section{Robust Subspace Recovery} \label{sec:robustsubspacerecovery}
\input{subspacerecovery.tex}

\section{Learning Hidden Markov Models}\label{sec:HMM}
\input{HMM.tex}

\section{Higher Order Tensor Decompositions}\label{sec:foobi}
\input{FOOBI.tex}

\vspace{15pt}

\appendix

\noindent {\LARGE \bf Appendix}

\input{appendix.tex}

\section*{Acknowledgements}

We would like to thank Anindya De for several helpful discussions, particularly those that led to the algorithm for robust subspace recovery. 

\bibliographystyle{alpha}
\bibliography{aravind}

\end{document}

%% file: marginnotes.tex
 \newif\ifnotes\notesfalse

\ifnotes
\usepackage{color}
\definecolor{mygrey}{gray}{0.50}
\newcommand{\notename}[2]{{\textcolor{mygrey}{\footnotesize{\bf (#1:} {#2}{\bf ) }}}}

\else

\newcommand{\notename}[2]{{}}

\fi

\newcommand{\anote}[1]{{\notename{Aravindan}{#1}}}
\newcommand{\bnote}[1]{{\notename{Aditya}{#1}}}
\newcommand{\pnote}[1]{{\notename{Aidan}{#1}}}
\newcommand{\cnote}[1]{{\notename{Aidao}{#1}}}

%% file: macros.tex
\newtheorem{theorem}{Theorem}[section]
\newtheorem*{theorem*}{Theorem}

\newtheorem{proposition}[theorem]{Proposition}
\newtheorem*{proposition*}{Proposition}
\newtheorem{lemma}[theorem]{Lemma}
\newtheorem*{lemma*}{Lemma}
\newtheorem{corollary}[theorem]{Corollary}
\newtheorem*{conjecture*}{Conjecture}

\newtheorem*{fact*}{Fact}

\newtheorem*{hypothesis*}{Hypothesis}

\newtheorem{itheorem}[theorem]{Informal Theorem}

\newtheorem*{claim*}{Claim}

\theoremstyle{definition}
\newtheorem{definition}[theorem]{Definition}

\theoremstyle{remark}
\newtheorem{remark}[theorem]{Remark}
\newtheorem*{remark*}{Remark}
\newtheorem{observation}[theorem]{Observation}
\newtheorem*{observation*}{Observation}

\newcommand{\eat}[1]{}

\newcommand{\R}{\mathbb{R}}

\newcommand{\Z}{\mathbb{Z}}

\newcommand{\calD}{\mathcal{D}}

\newcommand{\calN}{\mathcal{N}}
\newcommand{\calS}{{S}}

\newcommand{\nvgap}{\vspace{-6pt}}
\newcommand{\poly}{\mathrm{poly}}

 \newcommand{\set}[1]{\{ #1 \} }

\newcommand{\norm}[1]{\lVert #1 \rVert}

\newcommand{\iprod}[1]{\langle#1\rangle}
\newcommand{\Iprod}[1]{\left\langle#1\right\rangle}

\newcommand{\Esymb}{\mathbb{E}}
\newcommand{\Psymb}{\mathbb{P}}

\DeclareMathOperator*{\E}{\Esymb}

\DeclareMathOperator*{\ProbOp}{\Psymb}

\renewcommand{\Pr}{\ProbOp}

\newcommand{\diag}{\text{diag}}

\newcommand{\tr}{\text{tr}}

\newcommand{\tA}{\tilde{A}}

\newcommand{\eps}{\varepsilon}
\renewcommand{\epsilon}{\varepsilon}

\newcommand{\bu}{\mathbf{u}}
\newcommand{\bD}{\mathbf{D}}
\newcommand{\wh}[1]{\widehat{#1}}

\newcommand{\ot}{\otimes}

\newcommand{\ta}{\widetilde{a}}

\newcommand{\tu}{\widetilde{u}}

\newcommand{\Proj}{\mathrm{Proj}}
\newcommand{\tOcal}{\widetilde{\mathcal{O}}}
\newcommand{\Ocal}{\mathcal{O}}

%% file: intro.tex


Several basic computational problems in unsupervised learning like learning probabilistic models, clustering and representation learning are intractable in the worst-case. Yet practitioners have had remarkable success in designing heuristics that work well on real-world instances. Bridging this disconnect between theory and practice is a major challenge for many problems in unsupervised learning and high-dimensional data analysis. 


\anote{4/4: shortened. maybe more?}
The paradigm of Smoothed Analysis~\cite{ST00} has proven to be a promising avenue when the algorithm has only a few isolated bad instances.  
Given {\em any instance} from the whole problem space (potentially the worst input), smoothed analysis gives good guarantees for most instances in a small neighborhood around it; this is formalized by small random perturbations of worst-case inputs.  This powerful beyond worst-case paradigm has 
been used to analyze the simplex algorithm for solving linear programs~\cite{ST00}, linear binary optimization problems like knapsack and bin packing~\cite{BV06}, multi-objective optimization~\cite{MO11}, local max-cut~\cite{lmc1,lmc2}, and supervised learning~\cite{KST09}.
Smoothed analysis gives an elegant way of interpolating between traditional average-case analysis and worst-case analysis by varying the size of the random perturbations. 


In recent years, smoothed analysis has been particularly useful in unsupervised learning and high-dimensional data analysis, where the hard instances often correspond to adversarial degenerate configurations. For instance, consider the problem of finding a low-rank decomposition of an order-$\ell$ tensor that can be expressed as $T \approx \sum_{i=1}^k a_i \otimes a_i \otimes \dots \otimes a_i$. It is NP-hard to find a rank-$k$ decomposition in the worst-case when the rank $k \ge 6n$\cite{Has90} (this setting where the rank $k \ge n$ is called the \emph{overcomplete} setting). On the other hand, when the factors of the tensor $\set{a_i}_{i \in [k]}$ are perturbed with some small amount of random Gaussian noise, there exist polynomial time algorithms that can successfully find a rank-$k$ decomposition with high probability even when the rank is $k=O(n^{\lfloor (\ell-1)/2 \rfloor})$~\cite{BCMV}. Similarly, parameter estimation for basic latent variable models like mixtures of spherical Gaussians has exponential sample complexity in the worst case
~\cite{MV10}; yet, polynomial time guarantees can be obtained using smoothed analysis, where the parameters (e.g., means for Gaussians) are randomly perturbed in high dimensions~\cite{HK12, BCMV, ABGRV14, GHK}.\footnote{In many unsupervised learning problems, the random perturbation to the parameters can not be simulated by perturbations to the input (i.e., samples from the mixture). Hence unlike binary linear optimization~\cite{BV06}, such smoothed analysis settings in learning are not limited by known NP-hardness and hardness of approximation results.} 
 Smoothed analysis results have also been obtained for other problems like overcomplete ICA~\cite{GVX14}, learning mixtures of general Gaussians~\cite{GHK},  
fourth-order tensor decompositions~\cite{MSS}, and recovering assemblies of neurons~\cite{ADMPSV18}.   
\anote{Should we mention the neural nets application? }


The technical core of many of the above smoothed analysis results involves analyzing the minimum singular value of certain carefully constructed random matrices with dependent entries. Let $\set{\widetilde{a}_1, \widetilde{a}_2, \dots, \widetilde{a}_k}$ be random (Gaussian) perturbations of the points $\set{a_1, \dots, a_k} \subset \R^n$ (think of the average length of the perturbation to be $\rho=1/\poly(n)$). Typically, these correspond to the unknown parameters of the probabilistic model that we are trying to learn.  Proving polynomial smoothed complexity bounds often boils down to proving an inverse polynomial lower bound on the least singular value of certain matrices (that depend on the algorithm), where every entry is a multivariate polynomial involving some of the perturbed vectors $\set{\widetilde{a}_1, \dots, \widetilde{a}_k}$. These bounds need to hold with a sufficiently small failure probability over the randomness in the perturbations. 
%

Let us now consider some examples to give a flavor of the statements that arise in applications.
\begin{itemize}
    \item In learning mixtures of spherical Gaussians via tensor decomposition, the key matrix that arises is the ``product of means'' matrix, in which the number of columns is $k$, the number of components in the mixture, and the $i$th column is the flattened tensor $\ta_i^{\otimes \ell}$, where $\ta_i$ is the mean of the $i$th component.
    \item In the so-called FOOBI algorithm for tensor decomposition (proposed by~\cite{Cardoso}, which we will study later), the complexity as well as correctness of the algorithm depend on a special matrix $\Phi$ being well conditioned. $\Phi$ has the following form: each column corresponds to a pair of indices $i,j \in [k]$, and the $(i,j)$th column is $\ta_i^{\ot 2} \otimes \ta_j^{\ot 2} - (\ta_i \ot \ta_j)^{\ot 2}$.
    \item In learning hidden Markov models (HMMs), the matrix of interest is one in which each column is a sum of appropriate monomials of the form $\ta_{i_1} \ot \ta_{i_2} \ot \dots \ot \ta_{i_{\ell}}$, where $i_1 i_2 \dots i_{\ell}$ correspond to {\em length-$\ell$ paths} in the graph being learned. \bnote{are they really paths?}
\end{itemize}

For many of the recent algorithms based on spectral and tensor decomposition methods (e.g., ones in~\cite{AMR09, AGHKT12}), one can write down matrices whose condition numbers determine the performance of the corresponding algorithms (in terms of running time, error tolerance etc.). While there is a general theory with broadly applicable techniques (sophisticated concentration bounds) to derive high confidence upper bounds on the maximum singular value of such dependent random matrix ensembles, there are comparatively fewer general tools for establishing lower bounds on the minimum singular value (this has more of an ``anti-concentration'' flavor), except in a few special cases such as tensor decompositions (using ideas like partitioning co-ordinates).

The high level question that motivates this paper is the following: {\em can we obtain a general characterization of when such matrices have a polynomial condition number with high probability?} For instance, in the first example, we may expect that as long as $k < \binom{n+\ell-1}{\ell}$, the matrix has an inverse polynomial condition number (note that this is $\ll n^{\ell}$ due to the symmetries). 


There are two general approaches to the question above.  The first is a characterization that follows from results in algebraic geometry (see~\cite{AMR09, Strassen1983Rank}). These results state that the matrix of polynomials either has a sub-matrix whose determinant is the identically zero polynomial, or that the matrix is {\em generically} full rank. This means that the set of $\{\ta_i\}$ that result in the matrix having $\sigma_{\min} = 0$ has measure zero. However, note that this characterization is far from being quantitative. For polynomial time algorithms, we typically need $\sigma_{\min} \ge 1/\poly(n)$ with high probability (this is because polynomial sample complexity often requires these algorithms to be robust to inverse polynomial error). A second approach is via known anti-concentration inequalities for polynomials (such as the Carbery-Wright inequality~\cite{CarberyWright}). In certain settings, these can be used to prove that each column must have at least a small non-zero component orthogonal to the span of the other columns (which would imply a lower bound on $\sigma_{\min}$). However, it is difficult to use this approach to obtain strong enough probability guarantees for the condition number.

Our main contributions are twofold. 
The first is to prove lower bounds on the least singular value for some broad classes of random matrix ensembles where the entries are low-degree multivariate polynomials of the entries of a given set of randomly perturbed vectors. The technical difficulty arises due to the correlations in the perturbations (as different matrix entries could be polynomials of the same ``base'' variables). We note that even in the absence of correlations, (i.e., if the entries are perturbed independently), analyzing the least singular value is non-trivial and has been studied extensively in random matrix theory (see \cite{Taobook, RudelsonV}). \bnote{rephrased}
 
Our second contribution is to leverage these results and prove new smoothed analysis guarantees for learning overcomplete hidden markov models, and design algorithms with improved bounds for overcomplete tensor decompositions and for robust subspace recovery.


\section{Our Results and Techniques} \label{sec:results}

\anote{4/4: removed some introductory paragraphs.}

\subsection{Lower bounds on the Least Singular Value.}
The first setting we consider is a simple yet natural one. Suppose we have $k$ independently perturbed vectors $\widetilde{a}_1, \dots, \widetilde{a}_k$, and suppose we have a matrix in which each column is a fixed polynomial function of precisely one of the variables. We give a sufficient condition under which $\sigma_k$ ($k$th largest singular value, or the least singular value here since there are only $k$ columns) of this matrix is at least inverse polynomial with high probability.


\begin{theorem}\label{thm:columnpolynomials}
Let $\ell \in \Z_+$ be a constant and let $f: \R^n \to \R^m$ be a map defined by $m$ homogeneous polynomials $\{f_i\}_{i=1}^m$ of degree $\ell$.  Suppose that
$$f_i(x)=\sum_{\substack{J=(j_1, \dots, j_\ell) \in [n]^\ell \\ j_1 \le j_2 \le \dots \le j_\ell}} U_i(j_1, \dots, j_\ell) x_{j_1} x_{j_2} \dots x_{j_\ell},$$
and let $U \in \R^{m \times {n+\ell-1 \choose \ell}}$ denote the matrix of coefficients, with $i^{th}$ row $U_i$ corresponding to $f_i$.  For vectors $a_1, a_2, \dots, a_k \in \R^{n}$, let $M_f (a_1, a_2, \dots, a_k)$ denote the $m \times k$ matrix whose $(i,j)$th entry is $f_i(a_j)$. 
Then for any set of vectors $\{a_i\}_{i=1}^k$, with probability at least $1-k \exp\big(-\Omega_\ell(\delta n) \big)$,
\begin{equation}\label{eq:columnpolynomials_intro}
    \sigma_k\Big( M_f(\widetilde{a}_1, \dots, \widetilde{a}_k) \Big) \ge \frac{\Omega_\ell(1)}{\sqrt{k}}\Big(\frac{\rho}{n}\Big)^\ell  \cdot \sigma_{k+\delta {n+\ell-1 \choose \ell}}(U),
\end{equation}
where $\widetilde{a_j}$ represents a random perturbation of $a_j$ with independent Gaussian noise $N(0,\rho^2/n)^n$.
\end{theorem}

\anote{4/4: changed statement}
\anote{4/4: removed Carbery-Wright generalization forward pointer.} 
To obtain a non-trivial bound, note that we need $\sigma_{k +\delta \binom{n+\ell-1}{\ell}}(U) >0$. Qualitatively, $\sigma_k (U)$ being $>0$ is unavoidable. But more interestingly, we will see that the second term is also necessary. In particular, we demonstrate that 
$\Omega(\delta n^{\ell})$ non-trivial singular values are necessary for the required concentration bounds even when $k=1$ (see Proposition~\ref{claim:counterexample} for details). \bnote{is the $k=1$ correct?} 
In this sense, Theorem~\ref{thm:columnpolynomials} gives an almost tight condition for the least singular value of the above random matrix ensemble to be non-negligible. 

For an illustration of Theorem~\ref{thm:columnpolynomials}, consider the simple vector-valued polynomial function $f(x)=x^{\otimes \ell} \in \R^{n^\ell}$ (the associated matrix $U$ essentially just corresponds to the identity matrix, with some repeated rows). If $\widetilde{a}_1, \dots, \widetilde{a}_k \in \R^n$ are randomly perturbed vectors, the above theorem shows that the least singular value of the matrix $M_f(\widetilde{a}_1, \dots, \widetilde{a}_k)$ is inverse polynomial with exponentially small failure probability, as long as $k \le (1-o(1)) {n+\ell-1 \choose \ell}$ (earlier results only establish this when $k$ is smaller by a $\exp(\ell)$ factor, because of partitioning co-ordinates). In fact, the above example will be crucial to derive improved smoothed polynomial time guarantees for robust subspace recovery even in the presence of errors (Theorem~\ref{thm:robustrecovery:maintheorem}). 

The next setting we consider is one where the $j$th column of $M$ does not depend solely on $\ta_j$, but on a small subset of the columns in $\set{\ta_1, \dots, \ta_k}$ in a structured form. 
Specifically, in the random matrix ensembles that we consider, each of the $R$ columns of the matrix depends on a few of the vectors in $a_1, \dots, a_k$ as a ``monomial'' in terms of tensor products i.e., each column is of the form $u_1 \otimes u_2 \otimes \dots \otimes u_\ell$ where $u_1, u_2, \dots, u_\ell \in \set{\widetilde{a}_1, \dots, \widetilde{a}_k}$. To describe our result here, we need some notation. For two monomials $u_1 \otimes \dots \otimes u_\ell$ and $v_1 \otimes \dots \otimes v_\ell$, we say that they disagree in $s$ positions if $u_i \ne v_i$ for exactly $s$ different $i \in [\ell]$.  For a fixed column $j \in [R], s \in \set{0, 1, \dots, \ell}$, let $\Delta_s(j)$ represent the number of other columns whose monomial disagrees with that of column $j$ in exactly $s$ positions, and let $\Delta_s=\max_{j \in [R]} \Delta_s(j)$. (Note that $\Delta_0 = 0$ and $\Delta_\ell \le R$ by default).   
\begin{theorem} \label{thm:indepmonomials}
Let $\{\widetilde{a}_1, \ldots, \widetilde{a}_k\} \subseteq \R^n$ be a set of $\rho$-perturbed vectors, let $\ell \in \Z_+$ be a constant, and let $M \in \R^{n^\ell \times R}$ be a matrix whose columns $M_1, \ldots, M_R$ are tensor monomials in $\{\widetilde{a}_i\}_{i \in [k]}$. Let $\Delta_s$ be as above for $s = 1, \ldots, \ell$. If
\begin{equation}
    \sum_{s=1}^\ell \Delta_s \cdot \left(\frac{n}{\ell}\right)^{\ell-s} \le c\left(\frac{n}{\ell}\right)^\ell
\end{equation}
for some $c \in (0,1)$, then $\sigma_{R}(M) > \Omega_\ell(1)\cdot (\rho/n)^\ell/\sqrt{R}$ with probability at least $1-\exp\big(-\Omega_\ell(1-c)n + \log R \big)$.
\end{theorem}
The above statement will be useful in obtaining smoothed polynomial time guarantees for learning overcomplete hidden markov models (Theorem~\ref{ithm:HMM}), and for higher order generalizations of the FOOBI algorithm of \cite{Cardoso} that gives improved tensor decomposition algorithms up to rank $k=n^{\lfloor \ell /2 \rfloor}$ for order $\ell$ tensors (Theorem~\ref{ithm:foobi}). In both these applications, the matrix of interest (call it $M'$) is not a monomial matrix per se, but we express its columns as linear combinations of columns of an appropriate monomial matrix $M$. Specifically, it turns out that $M'=MP$, and $P$ has full column rank (in a robust sense). For example, in the case of overcomplete HMMs, each column of $M'$ is a sum of monomial terms of the form $\ta_{i_1} \ot \ta_{i_2} \ot \dots \ot \ta_{i_{\ell}}$, where $i_1 i_2 \dots i_{\ell}$ correspond to {\em length-$\ell$ paths} in the graph being learned. Each term corresponding to a length-$\ell$ path only shares dependencies with other paths that share a vertex.   
\anote{4/4: removed the ``other'' statement, which we decided to remove from the paper.}


\nvgap
\paragraph{Failure probability.}

The theorems above emphasize the dependence on the failure probability. We ensure that the claimed lower bounds on $\sigma_{\min}$ hold with a sufficiently small failure probability, say $n^{-\omega(1)}$ or even exponentially small (over the randomness in the perturbations). This is important because in smoothed analysis applications, the failure probability essentially describes the fraction of points around any given point that are {\em bad} for the algorithm. 
In many of these applications, the time/sample complexity, or the amount of error tolerance (as in the robust subspace recovery application we will see) has an inverse polynomial dependence on the minimum singular value. Hence, if we have a guarantee that $\sigma_{\min} \ge \gamma$ with probability $\ge 1- \gamma^{1/2}$ (as is common if we apply methods such as the Carbery-Wright inequality), we have that the probability of the running time exceeds $T$ (upon perturbation) is $\le 1/\sqrt{T}$. Such a guarantee does not suffice to show that the expected running time is polynomial (also called polynomial smoothed complexity).

\subsubsection{Techniques} \label{sec:techniques}
Theorem~\ref{thm:columnpolynomials} crucially relies on the following theorem, which may also be of independent interest.  
\begin{itheorem} \label{ithm:sym}
Let $V_\ell$ be the space of all symmetric order $\ell$ tensors in $\R^{n \times n \times \dots \times n}$, and let $\calS \subset V_\ell$ be an arbitrary subspace of dimension $(1-\delta) {n+1-\ell \choose \ell}$, for some $0<\delta < 1$. Let $\Pi^{\perp}_S$ represents the projection matrix onto the subspace of $V_\ell$ orthogonal to $\calS$.  Then for any vector $x$ and its $\rho$-perturbation $\widetilde{x}$, we have that
$\norm{\Pi^{\perp}_S \widetilde{x}^{\otimes \ell}}_2 \ge 1/\poly_\ell(n,1/\rho)$ with probability at least  $1-\exp\big( - \Omega_\ell(\delta n)\big)$.
\end{itheorem}

The proofs of the theorems above 
use as a black-box the smoothed analysis result of Bhaskara et al.~\cite{BCMV} and the improvements in Anari et al.~\cite{ADMPSV18} which shows minimum singular value bounds (with exponentially small failure probability) for tensor products of vectors that have been independently perturbed. Given $\ell \times k$ randomly perturbed vectors $\set{\widetilde{a}^{(j)}_i: j \in [\ell], i \in [k]}$, existing results~\cite{BCMV,ADMPSV18} analyze the minimum singular value of a matrix $M$ where the $i$th column ($i \in [k]$) is given by $\widetilde{a}_i^{(1)} \otimes \widetilde{a}_i^{(2)} \otimes \dots \otimes \widetilde{a}_i^{(\ell)}$. However this setting does not suffice for proving Theorem~\ref{thm:columnpolynomials}, Theorem~\ref{thm:indepmonomials}, or the different applications presented here because existing results assume the following two conditions:
\begin{enumerate}
\item The perturbations to the $\ell$ factors of the $i$th column i.e., $a_i^{(1)}, \dots, a_i^{(\ell)}$ are independent. For proving Theorem~\ref{thm:columnpolynomials} (and for Theorem~\ref{ithm:subspacerecovery}) we need to analyze symmetric tensor products of the form $\widetilde{x}_i^{\otimes \ell}$, where the perturbations across the factors are the same.
\item Each column of $M$ depends on a disjoint set of vectors $\widetilde{a}_i^{(1)}, \dots, \widetilde{a}_i^{(\ell)}$, i.e., any vector $\widetilde{a}_i^{(j)}$ is involved in only one column. For proving Theorem \ref{thm:indepmonomials} (and later in Theorems~\ref{ithm:HMM} and \ref{ithm:foobi}) however, the same perturbed vector may appear in several columns of $M$.
\end{enumerate}

Our main tool for proving Theorem~\ref{thm:columnpolynomials}, and Theorem~\ref{thm:indepmonomials} are various decoupling techniques to overcome the dependencies that exists in the randomness for different terms. Decoupling inequalities~\cite{dlP} are often used to prove concentration bounds (bounds on the upper tail) for polynomials of random variables. However, in our case they will be used to establish lower bounds on the minimum singular values. This has an anti-concentration flavor, since we are giving an upper bound on the ``small ball probability'' i.e., the probability that the minimum singular value is close to a small ball around $0$. For Theorem~\ref{thm:columnpolynomials} (and Theorem~\ref{ithm:sym}) which handles symmetric tensor products, we use a combination of asymmetric decoupling along with a positive correlation inequality for polynomials that is inspired by the work of Lovett~\cite{Lovett}. 

\bnote{The next paragraph overlapped significantly with old stuf.. so I commented out parts}
We remark that one approach towards proving lower bounds on the least singular value for the random matrix ensembles that we are interested in, is through a direct application of anti-concentration inequalities for low-degree polynomials like the Carbery-Wright inequality (see \cite{ABGRV14} for smoothed analysis bounds using this approach). Typically this yields an $\eps = 1/\poly(n)$ lower bound on $\sigma_{\min}$ with probability $\eps^{1/\ell}$ (where $\ell$ is the degree). As we observed above, this cannot lead to polynomial smoothed complexity for many problems.

\anote{4/4: Added. Rephrase or OK? }

\anote{4/5: rephrased below}
Interestingly we prove along the way, a vector-valued version of the Carbery-Wright anti-concentration inequality~\cite{CarberyWright, NSV03} (this essentially corresponds to the special case of Theorem~\ref{thm:columnpolynomials} when $k=1$). In what follows, we will represent a homogenous degree $\ell$ multivariate polynomial $g_j: \R^n \to \R$ using the symmetric tensor $M_j$ of order $\ell$ such that $g_j(x)= \iprod{M_j , x^{\otimes \ell}}$ (please see Section~\ref{sec:prelims} for the formal notation). 

\begin{itheorem} \label{ithm:vectorCW}
Let $\epsilon, \delta \in (0,1)$, $\eta>0$, and let $g: \R^{n} \to \R^m$ be a vector-valued degree $\ell$ homogenous polynomial of $n$ variables given by $g(x)=(g_1(x), \dots, g_m(x))$ such that the matrix $M \in \R^{m \times n^{\ell}}$, with the $i$th row being formed by the co-efficients of the polynomial $g_i$, has $\sigma_{\delta n^{\ell}} (M) \ge \eta$. Then for any fixed $u \in \R^n, t \in \R^m$, and $x \sim N(0,\rho^2/n)^n$ we have
\begin{align}
    \Pr\Big[ \norm{g(u+x) - t}_2 < \Omega_\ell( \epsilon \eta) \cdot \Big(\frac{\rho^{\ell}}{n^\ell}\Big) \Big]   &< \epsilon^{\Omega_\ell(\delta n)}.
\end{align}
\end{itheorem}
See Theorem~\ref{thm:newvectorCW} for a more formal statement. The main feature of the above result is that while we lose in the ``small ball'' probability with the degree $\ell$, we gain an $m^{\Omega(1)}$ factor in the exponent on account of having a vector valued function. The interesting setting of parameters is when $\ell=O(1),\rho=1/\poly(n), \epsilon=\poly_\ell(\rho/n)$ and  $\delta=n^{-o(1)}$. We remark that the requirement of $\delta n^{\ell}$ non-trivial singular values is necessary, as described in Proposition~\ref{claim:counterexample}.

The second issue mentioned earlier about \cite{BCMV, ADMPSV18} is that in many applications each column depends on many of the same underlying few ``base'' vectors. Theorem~\ref{thm:indepmonomials} identifies a simple condition in terms of the amount of overlap between different columns that allows us to prove robust linear independence for very different settings like learning overcomplete HMMs and higher order versions of the FOOBI algorithm. Here the decoupling is achieved by building on the ideas in \cite{MSS}, by carefully defining appropriate subspaces where we can apply the existing results on decoupled tensor products~\cite{BCMV,ADMPSV18}.

We now describe how the above results give new smoothed analysis results for three different problems in unsupervised learning.  

\subsection{Robust Subspace Recovery}

Robust subspace recovery is a basic problem in unsupervised learning where we are given $m$ points $x_1, \dots, x_m \in \R^n$, an $\alpha \in (0,1)$ fraction of which lie on (or close to) a $d$-dimensional subspace $T$. When can we find the subspace $T$, and hence the ``inliers'', that belong to this subspace? This problem is closely related to designing a robust estimator for subspace recovery: a $\beta$-robust estimator for subspace recovery approximately recovers the subspace even when a $\beta$ fraction of the points are corrupted arbitrarily (think of $\beta=1-\alpha$). 
The largest value of $\beta$ that an estimator tolerates is called the breakdown point of the estimator. This problem has attracted significant attention in the robust statistics community~\cite{Rou84,RL05,DH83}, yet many of these estimators are not computationally efficient in high dimensions. On the other hand, the singular value decomposition is not robust to outliers. Hardt and Moitra~\cite{HM13} gave the first algorithm for this problem that is both computationally efficient and robust. Their algorithm successfully estimates the subspace $T$ when $\alpha> d/n$, assuming a certain non-degeneracy condition about both the inliers and outliers.\footnote{This general position condition holds in a smoothed analysis setting} This algorithm is also robust to some small amount of noise in each point i.e., the inliers need not lie exactly on the subspace $T$. They complemented their result with a computational hardness in the worst-case (based on the Small Set Expansion hypothesis) for finding the subspace when $\alpha< d/n$. 

We give a simple algorithm that for any constants $\ell \ge 1, \delta>0$ runs in $\poly(mn^{\ell})$ time and in a smoothed analysis setting, provably recovers the subspace $T$ with high probability, when $\alpha \ge (1+\delta) (d/n)^\ell$. 
Note that this is significantly smaller than the bound of $(d/n)$ from~\cite{HM13} when $\ell > 1$. 
For instance in the setting when $d=(1-\eta)n$ for some constant $\eta >0$ (say $\eta=1/2$), our algorithms recovers the subspace when the fraction of inliers is {\em any constant} $\alpha>0$ by choosing $\ell = O(\log(\alpha)/\log(1-\eta))$, while the previous result requires that at least $\alpha > 1 - \eta$ of the points are inliers. On the other hand, when $d/n = n^{-\Omega(1)}$ the algorithm can tolerate any inverse polynomially small $\alpha$, in polynomial time.   
In our smoothed analysis setting, each point is given a small random perturbation  -- each outlier is perturbed with a $n$-variate Gaussian $N(0,\rho^2)^n$ (think of $\rho=1/\poly(n)$), and each inlier is perturbed with a projection of a $n$-variate Gaussian $N(0,\rho^2)^n$ onto the subspace $T$. Finally, there can be some adversarial noise added to each point (this adversarial noise can in fact depend on the random perturbations).     
\anote{4/4: edited mildly. Reread. }

\begin{itheorem}\label{ithm:subspacerecovery}
For any $\delta\in(0,1), \ell\in \Z_+$ and $\rho>0$. 
Suppose there are $m = \Omega( n^{\ell}+ d /(\delta \alpha))$ points $x_1, \dots, x_m \in \R^n$ which are randomly $\rho$-perturbed according to the smoothed analysis model described above, with an $\alpha\geq (1+\delta)\binom{d+\ell-1}{\ell}/\binom{n+\ell-1}{\ell}$ fraction of the points being inliers, and total adversarial noise $\epsilon_0 \le \poly_\ell(\rho/m)$. Then there is an efficient algorithm that returns a subspace $T'$ with $\norm{sin\Theta(T,T')}_F\leq \poly_\ell(\epsilon_0,\rho, 1/m)$ with probability at least $1-\exp\big(-\Omega_\ell(\delta n)+ 2\log m \big) - \exp(-\Omega(d\log m))$.
\end{itheorem}

See Section~\ref{sec:robustsubspacerecovery} for a formal statement, algorithm and proof. 
While the above result gives smoothed analysis guarantees when $\alpha$ is at least $(d/n)^{\ell} < d/n$,  the hardness result of \cite{HM13} shows that finding a $d$-dimensional subspace that contains an $\alpha<d/n$ fraction of the points is computationally hard assuming the Small Set Expansion conjecture. Hence our result presents a striking contrast between the intractability result in the worst-case and a computationally efficient algorithm in a smoothed analysis setting when $\alpha > (d/n)^{\ell}$ for some constant $\ell\ge 1$.  
Further, we remark that the error tolerance of the algorithm (amount of adversarial error $\eps_0$) does not depend on the failure probability. \anote{4/4: Mention this in Section as well} 

\paragraph{Techniques and comparisons.}

The algorithm for robust subspace recovery at a high level follows the same approach as Hardt and Moitra~\cite{HM13}. Their main insight was that if we sample a set of size slightly less than $n$ from the input, and if the fraction of inliers is $> (1+\delta) d/n$, then there is a good probability of obtaining $>d$ inliers, and thus there exist points that are in the linear span of the others. Further, since we sampled fewer than $n$ points and the outliers are also in general position, one can conclude that the only points that are in the linear span of the other points are the inliers.

Our algorithm for handling smaller $\alpha$ is simple and is also tolerant to an inverse polynomial amount of adversarial noise in the points. Our first observation is that we can use a similar idea of looking for linear dependencies, but with tensored vectors!
Let us illustrate in the case $\ell = 2$. Suppose that the fraction of inliers is $> (1+\delta) \binom{d+1}{2} / \binom{n+1}{2}$. Suppose we take a sample of size slightly less than  $\binom{n+1}{2}$ points from the input, and consider the flattened vectors $x \ot x$ of these points. As long as we have more than $\binom{d+1}{2}$ inliers, we expect to find linear dependencies among the tensored inlier vectors. However, we need to account for the adversarial error in the points (this error could depend on the random perturbations as well). For each point, we will look for ``bounded'' 
linear combinations that are close to the given point.  
Using  Theorem~\ref{ithm:sym}, we can show that such dependencies cannot involve the outliers.  
This in turn allows us to recover the subspace even when $\alpha> (d/n)^{\ell}$ for any constant $\ell$ in a smoothed analysis sense. 

We remark that the earlier least singular value bounds of \cite{BCMV} can be used to show a weaker guarantee about robust linear independence of the matrix formed by columns $\tilde{x}_i^{\otimes \ell}$ with a $c^{\ell}$ factor loss in the number of columns (for a constant $c \approx e$). This translates to an improvement over \cite{HM13} only in the regime when $d< n/c$. Our tight characterization in Theorem~\ref{ithm:sym} is crucial for our algorithm to beat the $d/n$ threshold of \cite{HM13} for any dimension $d <n$. 

Secondly, if there is no adversarial noise added to the points, it is possible to use weaker concentration bounds (e.g., Carbery-Wright inequality). In this case, our machinery is not required (although to the best of our knowledge, even an algorithm for this noise-free regime with a breakdown point $<d/n$ was not known earlier). In the presence of noise, the weaker concentration inequalities require a noise bound that is tied to the intended failure probability of the algorithm in a strong way. Using Theorem~\ref{ithm:sym} allows us to achieve a large enough adversarial noise tolerance $\eps_0$, that does not affect the failure probability of the algorithm. 
\anote{4/4: mention the weak anticoncentration based statement.}
\bnote{Changed slightly; you might want to see if reverting makes sense.}

\subsection{Learning Overcomplete Hidden Markov Models}

Hidden Markov Models (HMMs) are latent variable models that are extensively used for data with a sequential structure, like reinforcement learning, speech recognition, image classification, bioinformatics etc \cite{Edd96,GY08}. In an HMM, there 
is a hidden state sequence $Z_1,Z_2,\dots,Z_m$ taking values in $[k]$, that forms a stationary Markov chain $Z_1 \rightarrow Z_2 \rightarrow \dots \rightarrow Z_m$ with transition matrix $P$ and initial distribution $w=\{w_j\}_{j \in [r]}$ (assumed to be the stationary distribution). The observation $X_t$ 
is represented by a vector in $x^{(t)} \in \R^n$. Given the state $Z_t$ at time $t$, $X_t$ (and hence $x^{(t)}$) is conditionally independent of all other observations and states. The matrix $\Ocal$ (of size $n \times r$) represents the probability distribution for the observations: the $i$th column $\Ocal_i \in \R^n$ represents the expectation of $X_t$ conditioned on the state $Z_t=i$ i.e.
$$\forall i \in [r], t \in [m],  \quad \E[X_t \vert Z_t = i]= \Ocal_i \in \R^n.$$
In an HMM with continuous observations, the distribution of the observation conditioned on state being $i$ can be a Gaussian and $i$th column of $\Ocal$ would correspond to its mean. 
In the discrete setting, each column of $\Ocal$ can correspond to the parameters of a discrete distribution over an alphabet of size $n$. 

An important regime for HMMs in the context of many settings in image classification and speech is the {\em overcomplete setting} where the dimension of the observations $n$ is much smaller than state space $r$. 
Many existing algorithms for HMMs are based on tensor decompositions, and work in the regime when $n\le r$~\cite{AHK12, AGHKT12}. In the overcomplete regime, there have been several works~\cite{AMR09, BCV, Huangetal} that establish identifiability (and identifiability with polynomial samples) under some non-degeneracy assumptions, but obtaining polynomial time algorithms has been particularly challenging in the overcomplete regime. Very recently Sharan et al.~\cite{Sharanetal} gave a polynomial time algorithm for learning the parameters of an overcomplete discrete HMMs when the observation matrix $M$ is random (and sparse), and the transition matrix $P$ is well-conditioned, under some additional sparsity assumptions on both the transition matrix and observation matrix (e.g.,  the degree of each node in the transition matrix $P$ is at most $n^{1/c}$ for some large enough constant $c >1$).  
Using Theorem~\ref{thm:indepmonomials}, we give a polynomial time algorithm in the more challenging smoothed analysis setting where entries of $M$ are randomly perturbed with small random Gaussian perturbations \footnote{While small Gaussian perturbations makes most sense in a continuous observation setting, we believe that these ideas should also imply similar results in the discrete setting for an appropriate smoothed analysis model.}.

\begin{itheorem}\label{ithm:HMM}
Let $\eta,\delta \in (0,1)$ be constants. Suppose we are given a Hidden Markov Model with $r$ states and with $n\ge r^{\eta}$ dimensional observations with hidden parameters $\widetilde{O},P$. Suppose the transition matrix $P$ is $d\le n^{1-\delta}$ sparse (both row and column) and $\sigma_\text{min}(P) \ge \gamma_1>0$, and the each entry of the observation matrix is $\rho$-randomly perturbed (in a smoothed analysis sense), and the stationary distribution $w \in [0,1]^r$ has $\min_{i \in [r]} w_i \ge \gamma_2>0$, then there is a polynomial time algorithm that uses samples of time window $\ell \le 1/(\eta \delta)$ and recovers the parameters up to $\epsilon$ accuracy (in Frobenius norm) in time $(n/(\rho\gamma_1\gamma_2\epsilon))^{O(\ell)}$, with probability at least $1-\exp\big( - \Omega_\ell(n)\big)$.
\end{itheorem}


For comparison, the result of Sharan et al.~\cite{Sharanetal} applies to discrete HMMs, and gives an algorithm that uses time windows of size $\ell=O(\log_n r)$ in time $\poly(n,r,1/\eps,1/\gamma_1, 1/\gamma_2)^{\ell}$ (there is no extra explicit lower bound on $n$).  But it assumes that the observation matrix $\Ocal$ is fully random, and has other assumptions about sparsity about both $\Ocal$ and $P$, 
and about non-existence of short cycles. On the other hand, we can handle the more general smoothed analysis setting for the observation matrix $\Ocal$ for $n=r^{\eta}$ (for any constant $\eta>0$), 
and assume no additional conditions about non-existence of short cycles.
\anote{11/28 added:} To the best of our knowledge, this gives the first polynomial time guarantees in the smoothed analysis setting for learning overcomplete HMMs. 

Our results complement the surprising sample complexity lower bound in Sharan et al.~\cite{Sharanetal} who showed that it is statistically impossible to recover the parameters with polynomial samples when $n=\text{polylog} (r)$, even when the observation matrix is random. The algorithm is based on an existing approach using tensor decompositions~\cite{AMR09, AGHKT12, BCMV, Sharanetal}. The robust analysis of the above algorithm (Theorem~\ref{ithm:HMM}) follows by a simple application of Theorem~\ref{thm:indepmonomials}.
\anote{4/5:rephrased.}

\subsection{Overcomplete Tensor Decompositions}
Tensor decomposition has been a crucial tool in many of the recent developments in showing learning guarantees for unsupervised learning problems.  The problem here is the following. Suppose $A_1, \dots, A_R$ are vectors in $\R^n$. Consider the $s$'th order moment tensor
\[ M_{s} = \sum_{i=1}^R A_i^{\ot s}. \]

\newcommand{\err}{\mathsf{Err}}
The question is if the decomposition $\{A_i\}$ can be recovered given access only to the tensor $M_s$. This is impossible in general. For instance, with $s = 2$, the $A_i$ can only be recovered up to a rotation.  The remarkable result of Kruskal~\cite{Kru77} shows that for $s > 2$, the decomposition in ``typically'' unique, as long as $R$ is a small enough. Several works~\cite{Har70, Cardoso, AGHKT12, MSS} have designed efficient recovery algorithms in different regimes of $R$, and assumptions on $\{A_i\}$. The other important question is if the $\{A_i\}$ can be recovered assuming that we only have access to $M_s + \err$, for some noise tensor $\err$. 

Works inspired by the sum-of-squares hierarchy achieve the best dependence on $R$ (i.e., handle the largest values of $R$), and also have the best noise tolerance, but require strong incoherence (or even Gaussian) assumptions on the $\{A_i\}$~\cite{BKS15, HSSS16}. Meanwhile, spectral algorithms (such as~\cite{GVX14, BCMV}) achieve a weaker dependence on $R$ and can tolerate a significantly smaller amount of noise, but they allow recoverability for smoothed vectors $\{A_i\}$, which is considerably more general than recoverability for random vectors. The recent work of~\cite{MSS} bridges the two approaches in the case $s=4$.

Our result here is a decomposition algorithm for $2\ell$'th order tensors that achieves efficient recovery guarantees in the smoothed analysis model, as long as $R \le c n^{\ell}$ for a constant $c$.  Our result is based on a generalization of the ``FOOBI algorithm'' of Cardoso~\cite{Cardoso, DLCC}, who consider the case $\ell = 2$.
We also give a robust analysis of this algorithm (both the FOOBI algorithm for $\ell=2$, and our generalization to higher $\ell$): we show that the algorithm can recover the decomposition to an arbitrary precision $\eps$ (up to a permutation), as long as $\norm{\err} \le \poly_{\ell} (\eps, 1/n, \rho)$, where $\rho$ is the perturbation parameter in the smoothed analysis model. 

\begin{itheorem}\label{ithm:foobi}
Let $\ell \ge 2$ be an integer. Suppose we are given a $2\ell$'th order tensor $T = \sum_{i=1}^R A_i^{\ot 2\ell} + \err$, where $A_i$ are $\rho$-perturbations of vectors with polynomially bounded length. Then with probability at least $1-\exp(-\Omega_\ell(n))$, we can find the $A_i$ up to any desired accuracy $\eps$ (up to a permutation), assuming that $R < c n^{\ell}$ for a constant $c = c(\ell)$, and $\norm{\err}_F$ is a sufficiently small polynomial in $\eps, \rho, 1/n$.
\end{itheorem}

See Theorem~\ref{thm:foobi-robust} and Section~\ref{sec:foobi} for a formal statement and details. We remark that there exists different generalizations of the FOOBI algorithm of Cardoso to higher $\ell>2$ ~\cite{DLCChigher}. However, to the best of our knowledge, there is no analysis known for these algorithms that is robust to inverse polynomial error. \anote{4/5: removed for $\ell>2$.} Further our new algorithm is a very simple generalization of Cardoso's algorithm to higher $\ell$.
\anote{11/28: why inverse poly robustness matters? Mentioned anywhere?}

This yields an improvement in the best-known dependence on the rank in such a smoothed analysis setting --- from $n^{\ell - 1}$ (from~\cite{BCMV}) to $n^{\ell}$. Previously such results were only known for $\ell =2$ in~\cite{MSS}, who analyzed an SoS-based algorithm that was inspired by the FOOBI algorithm (to the best of our knowledge, their results do not imply a robust analysis of FOOBI). Apart from this quantitative improvement, our result also has a more qualitative contribution: it yields an algorithm for the problem of finding symmetric rank-1 tensors in a linear subspace.

\begin{itheorem}\label{ithm:rank1-recovery}
Suppose we are given a basis for an $R$ dimensional subspace $S$ of $\R^{n^\ell}$ that is equal to the span of the flattenings of $A_1^{\ot \ell}, A_2^{\ot \ell}, \dots A_R^{\ot \ell}$, where the $A_i$ are unknown $\rho$-perturbed vectors.  Then the $A_i$ can be recovered in time $\poly_{\ell}(n, 1/\rho)$ with probability at least $1-\exp(-\Omega_\ell(n))$.  Further, this is also true if the original basis for $S$ is known up to an inverse-polynomial perturbation.
\end{itheorem}

\anote{4/5:Removed above due to repetition.}

\paragraph{Techniques.}
At a technical level, the FOOBI algorithm of~\cite{Cardoso, DLCC} for decomposing fourth-order tensors rests on a rank-1 detecting ``device'' $\Phi$ that evaluates to zero if the inputs are a symmetric product vector, and is non-zero otherwise. We construct such a device for general $\ell$, and further analyze the condition number of an appropriate matrix that results using Theorem~\ref{thm:indepmonomials}.

We also give an analysis of the robustness of the FOOBI algorithm of~\cite{Cardoso}
and our extension to higher $\ell$. While such robustness analyses are often straightforward, and show that each
of the terms estimated in the proofs will be approximately preserved.
In the case of the FOOBI algorithm, this turns out to be impossible to
do (this is perhaps one reason why proving robust guarantees for the FOOBI algorithm even for $\ell=2$ has been challenging) . The reason is that the algorithm involves finding the top $R$
eigenvectors of the flattened moment matrix, and setting up a linear
system of equations in which the coefficients are {\em non-linear}
functions of the entries of the eigenvectors. Now, unless each of the
eigenvectors is preserved up to a small error, we cannot conclude that
the system of equations that results is close to the one in the
noise-free case. Note that for eigenvectors to be preserved
approximately after perturbation, we will need to have {\em sufficient
gaps} in the spectrum to begin with. This turns out to be impossible
to guarantee using current smoothed analysis techniques. We thus need
to develop a better understanding of the solution to the linear
system, and eventually argue that even if the system produced is quite
different, the solution obtained in the end is close to the original.


%% file: decoupling.tex

\section{Decoupling and Symmetric Tensor Products}

In this section we prove Theorem~\ref{thm:columnpolynomials} and related theorems about the least singular value of random matrices in which each column is a function of a single random vector. The proof of Theorem~\ref{thm:columnpolynomials} relies on the following theorem which forms the main technical theorem of this section. 
\begin{theorem}[Same as Theorem~\ref{ithm:sym}]\label{thm:newdecoupling}
Let $\delta \in (0,1)$, and let $V_\ell$ be space of all symmetric order $\ell$ tensors in $\R^{n \times n \times \dots \times n}$ (dimension is $D={ n+\ell-1 \choose \ell}$ ), and let $W \subset V_\ell$ be an arbitrary subspace of dimension $\delta D$. Then we have for any $x \in \R^n$ and $\tilde{x}=x+z$ where $z \sim N(0,\rho^2/n)^n$
$$\Pr_z\Big[\norm{\Pi_W \tilde{x}^{\otimes \ell}}_2 \ge\frac{c_1(\ell)\rho^\ell}{n^\ell}   \Big]\ge 1-\exp\Big( - c_2(\ell) \delta  n \Big) ,$$
where $c_1(\ell), c_2(\ell)$ are constants that depend only on $\ell$.
\end{theorem}
\cnote{03.04:use ADMPSV to improve the result, we also need to modify Theorem~\ref{ithm:sym}}

Theorem~\ref{thm:columnpolynomials} follows by combining the above theorem with an additional lemma that uses a robust version of Sylvester's inequality for products of matrices (see Section~\ref{sec:sym:finishingproof}). 
%
Our main tool will be the idea of decoupling, along with Lemma~\ref{lem:quantitativebound} that handles tensor products of vectors that have been perturbed independently. 
While decoupling inequalities~\cite{dlP} are often used to prove concentration bounds for polynomials of random variables, here this will be used to establish lower bounds on projections and minimum singular values, which have more of an anti-concentration flavor.  

In fact we can use the same ideas to prove the following anti-concentration statement that can be seen as a variant of the well-known inequality of Carbery and Wright~\cite{CarberyWright, NSV03}. In what follows, we will represent a degree $\ell$ multivariate polynomial $g_j: \R^n \to \R$ using the symmetric tensor $M_j$ of order $\ell$  such that $g_j(x)= \iprod{M_j , x^{\otimes \ell}}$. 

\anote{11/3: Need an extra condition that $m \le n^{\ell}$; otherwise $m$ can be arbitrarily large. We can also reword it as saying for any $m \ge \delta n^{\ell}$ if $\sigma_m(M)$ is non-negligible, then we have the correct bound. }
\begin{theorem} \label{thm:newvectorCW}
Let $\epsilon, \delta \in (0,1)$, $\eta>0$ and let $\ell \ge 2$ be an integer. Let $g: \R^{n} \to \R^m$ be a vector-valued degree $\ell$ homogenous polynomial of $n$ variables given by $g(x)=(g_1(x), \dots, g_m(x))$ where for each $j \in [m]$,  $g_j(x)=\iprod{M_j, x^{\otimes \ell}}$ for some symmetric order $\ell$ tensor $M_j \in \R^{n^{\ell}}$. Suppose the matrix $M \in \R^{m \times n^{\ell}}$ formed with the $(M_j:j \in [m])$ as rows has $\sigma_{\delta n^{\ell}} (M) \ge \eta$,  
then for any fixed $u \in \R^n, t \in \R^m$, and $z \sim N(0,\rho^2/n)^n$ we have
\begin{align}
    \Pr\Big[ \norm{g(u+z) - t}_2 < c(\ell)\epsilon \eta \cdot \Big(\frac{ \rho^\ell}{n^\ell}\Big)  \Big] &< \epsilon^{c'(\ell)\delta n}
\end{align}
where $c(\ell),c'(\ell)>0$ are constants that depend only on $\ell$. 
\end{theorem}
\anote{4/2: Check the modified remarks}
\begin{remark}
{\em Comparison to Carbery-Wright inequality:} Anti-concentration inequalities for polynomials are often stated for a single polynomial. They take the following form: if $g : \R^n \rightarrow \R$ is a degree-$\ell$ polynomial with $\norm{g}_2 \ge \eta$, and $x \sim N(0,1)^n$ (or other distributions like the uniform measure on a convex body), then the probability that $$\Pr_{x \sim N(0,1)^n}\Big[|g(x) - t| < \eps \eta \Big]\leq O(\ell) \cdot \eps^{1/\ell} . $$
\cnote{4.4:this statement seems incorrect, if the coefficients of the polynomials are extremely small, it will concentrate around 0?}\anote{4/5: Thanks. Fixed I think.}
    Our statement in Theorem~\ref{thm:newvectorCW} applies to {\em vector valued} polynomials $g$. Here, if the $g_j$ are ``different enough'', one can hope that the dependence above becomes $O(\eps^{m/\ell})$, where $m$ is the number of polynomials. Our statement may be viewed as showing a bound that is qualitatively of this kind (albeit with a much weaker dependence on $\ell$, when $\ell\geq 2$), when $m \ge \delta n^{\ell}$. We capture the notion of $g_j$ being different using the condition on the singular value of the matrix $M$. We also note that the paper of Carbery and Wright~\cite{CarberyWright} does indeed consider vector-valued polynomials, but their focus is on obtaining $\eps^{1/\ell}$ type bounds with a better constant for $\eps$. To the best of our knowledge, none of the known results try to get an {\em advantage} due to having multiple $g_j$.
\end{remark}
\begin{remark}
While the condition of $\delta n^{\ell}$ non-negligible singular values seems strong, this in fact turns out to be necessary. 
Proposition~\ref{claim:counterexample} shows that the relation between the failure probability and the number of non-negligible singular values is tight up to constants that depend only on $\ell$. In fact, $m \ge n^{\ell-1}$ is necessary to get any non-trivial bounds. Getting a tight dependence in the exponent in terms of $\ell$ is an interesting open question.
\end{remark}

The main ingredient in the proof of the above theorems is the following decoupling inequality. 
\begin{proposition}\label{prop:decoupling}[Anticoncentration through Decoupling]
Let $\epsilon >0$ and let $\ell \ge 2$ be an integer, and let $\norm{\cdot}$ represent any norm over $\R^n$. Let $g: \R^{n} \to \R^m$ be given by $g(x)=(g_1(x), \dots, g_m(x))$ where for each $j \in [m]$,  $g_j(x):=\iprod{M_j,x^{\otimes\ell}}$  is a multivariate homogeneous polynomial of degree $\ell$, and $M_j$ is a symmetric tensor of order $\ell$. For any fixed $u \in \R^n, t \in \R^m$, and $z \sim N(0,\rho^2)^n$ we have
\begin{align}
    \Pr_z\Big[ \norm{g(u+z) - t} \le \epsilon \Big] &\le \Pr\Big[ \norm{\widehat{g}(u+z_0, z_1,z_2,\dots, z_{\ell-1})} \le  \epsilon/ \ell!  \Big]^{1/2^{\ell-1}}, \\
    \text{where } \forall j \in [m], ~ \widehat{g}_j(u+z_0, z_1, \dots, z_{\ell-1})&= \iprod{M_j, (u+z_0) \otimes z_1 \otimes \dots \otimes z_{\ell-1}}, \nonumber 
\end{align}
$z_0 \sim N(0,\rho^2 (\ell+1)/(2\ell))^n$,
and $z_1, z_2, \dots, z_{\ell-1} \sim N(0,\rho^2/(2\ell))^n$.
\end{proposition}
\cnote{03.10:change $\epsilon\in(0,1)$ to $\epsilon>0$}

Note that in the above proposition, the polynomials $\widehat{g}_j(z_0, z_1, \dots, z_{\ell-1})=\iprod{M_j, (u+z_0) \otimes z_1 \otimes \dots \otimes z_{\ell-1}}$ correspond to decoupled multilinear polynomials of degree $\ell$. Unlike standard decoupling statements, here the different components $u+z_0, z_1, \dots, z_{\ell-1}$ are not identically distributed. 
We also note that the proposition itself is inspired by a similar lemma in the work of Lovett~\cite{Lovett} on an alternate proof of the Carbery-Wright inequality. Indeed the basic inductive structure of our argument is similar (going via Lemma~\ref{lem:positivecorrelation} below), but the details of the argument turn out to be quite different. In particular we want to consider a random perturbation around an arbitrary point $u$,\anote{4/2: Added previous sentence.} and moreover the proposition above deals with vector-valued polyomials $g$, as opposed to real valued polynomials in~\cite{Lovett}.

Theorem~\ref{thm:newdecoupling} follows by combining Proposition~\ref{prop:decoupling} and the theorem for decoupled tensor products (Lemma~\ref{lem:quantitativebound}).  This will be described in Section~\ref{sec:proof-decoupling}. 
Later in Section~\ref{sec:combinatorialproof}, we also give an alternate simple proof of  Theorem~\ref{thm:newdecoupling} for $\ell=2$ that is more combinatorial.  First we introduce the slightly more general setting for decoupling that also captures the required smoothed analysis statement. 

\subsection{Proof of Proposition~\ref{prop:decoupling}}

We will start with a simple fact involving signed combinations. 
\begin{lemma}\label{lem:randomsigns}
Let $\alpha_0, \alpha_1, \dots, \alpha_m$ be real numbers, and let $\zeta_1, \zeta_2, \dots, \zeta_m \in \set{\pm 1}$ be independent Rademacher random variables. Then
$$ \E_{\zeta}\Big[  \big( \alpha_0+ \alpha_1 \zeta_1 + \dots +\alpha_m \zeta_m \big)^{m+1} \prod_{i \in [m]} \zeta_i \Big] = (m+1)! \cdot \alpha_0 \alpha_1 \dots \alpha_m . $$
\end{lemma}
\begin{proof}
For a subset $S \subseteq [m]$, let $\xi(S)=\prod_{i \in S} \zeta_i$. Then it is easy to check that $\E\big[ \xi(S) \prod_{i \in [m]}\zeta_i  \big]=0$ if $S \ne \set{1, 2, \dots,m}$, and $1$ if $S=[m]$. 
Applying this along with the multinomial expansion for $\big( \alpha_0+ \alpha_1 \zeta_1 + \dots +\alpha_m \zeta_m \big)^{m}$ gives the lemma.  
\end{proof}

\begin{lemma}\label{lem:deg-l}
Consider any symmetric order $\ell$ tensor $T$, a fixed vector $x \in \R^n$ , and let $z_1 \sim N(0,\rho_1^2)^n, \dots, z_\ell \sim N(0,\rho_\ell^2)^n$ be independent random Gaussians. Then we have
\begin{align}
\sum_{\zeta_2, \dots, \zeta_\ell \in \pm 1} \Big(\prod_{i=2}^\ell \zeta_i \Big) \iprod{T, (x+z_1 + \zeta_2 z_2 + \dots+\zeta_\ell z_\ell)^{\otimes \ell}} &= 2^{\ell -1}   \ell ! \cdot \iprod{T, (x+z_1) \otimes z_2 \otimes \dots \otimes z_\ell} \label{eq:decoupling:eq}.
\end{align}
\end{lemma}
Note that the right side corresponds to the evaluation of the tensor $T$ at a random perturbation of $(x,0,0,\dots, 0)$.
\begin{proof}
First, we observe that since $T$ is symmetric, it follows that $\iprod{T, u_1 \otimes u_2\otimes \dots \otimes u_\ell}=\iprod{T, u_{\pi(1)} \otimes u_{\pi(2)}\otimes \dots \otimes u_{\pi(\ell)}}$ for any permutation $\pi$ on $(1,2,\dots,\ell)$. Let $u=x+z_1$, and let $\zeta_2, \zeta_3, \dots, \zeta_{\ell} \in \set{\pm 1}$ be independent Rademacher random variables. For any symmetric decomposition into rank-one tensors $T=\sum_j \lambda_j v_j^{\otimes \ell}$ (note that such a decomposition always exists for a symmetric tensor; see ~\cite{Comonetal} for example), we have for every $x \in \R^n$, 
$\iprod{T, x^{\otimes \ell}}=\sum_j \lambda_j \iprod{v_j, x}^{\ell}$. Applying Lemma~\ref{lem:randomsigns} (with $m=\ell-1$) to each term separately 
$$ \forall j, ~ \E_{\zeta_2, \zeta_3, \dots, \zeta_\ell} \Big[\Big(\prod_{i=2}^\ell \zeta_i \Big)  \iprod{v_j^{\otimes\ell}, (u+\zeta_2 z_2 +\dots+\zeta_\ell z_\ell)^{\otimes \ell}} \Big] = \ell ! \cdot \iprod{v_j^{\otimes \ell},u\otimes\zeta_2\otimes\cdots\otimes\zeta_\ell}.$$

Combining them, we get
\begin{align*}
   \E_{\zeta_2, \zeta_3, \dots, \zeta_\ell} \Big[\Big(\prod_{i=2}^\ell \zeta_i \Big)  \iprod{T, (u+\zeta_2 z_2 +\dots+\zeta_\ell z_\ell)^{\otimes \ell}} \Big]&= \ell ! \cdot \iprod{T, u \otimes z_2 \otimes z_3 \otimes \dots \otimes z_\ell}\\
   &= \ell ! \cdot \iprod{T, (x+z_1) \otimes z_2 \otimes z_3 \otimes \dots \otimes z_\ell}.
\end{align*}


\end{proof}

Our proof of the anti-concentration statement (Proposition~\ref{prop:decoupling}) will rely on the signed combination of vectors given in Lemma~\ref{lem:deg-l} and on a positive correlation inequality that is given below.

\begin{lemma}\label{lem:positivecorrelation}
Let $z \sim N(0,\rho^2)^n$ be an n-variate Gaussian random variable, and let $z_0 \sim N(0,\rho^2 (\ell+1)/(2\ell))^n$ and $z_1, z_2, \dots, z_{\ell-1} \sim N(0,\rho^2/(2\ell))^n$ be a collection of independent n-variate Gaussian random variables. Then for any measurable set $S \subset \R^n$ we have
\begin{equation}
    \Pr_z\Big[ z \in S \Big] \le \Pr_z\Big[ \bigwedge_{\zeta_1, \dots, \zeta_{\ell-1} \in \set{\pm 1}} \big( z_0+ \textstyle\sum_{j=1}^{\ell-1} \zeta_j z_j\big)  \in S  \Big]^{1/(2^{\ell-1})}
\end{equation}
\end{lemma}
This inequality and its proof are inspired by the work of Lovett~\cite{Lovett} mentioned earlier. The main  advantage in our inequality is that the right side here involves the particular signed combinations of the function values at $2^{\ell-1}$ points from $\ell$ independent copies that directly yields the asymmetric decoupled product (using Lemma~\ref{lem:deg-l}). 

\begin{proof}
Let $x_0, x_1, \dots, x_{\ell-1} \sim N(0,\rho^2/\ell)^n$, and for each $k \in [\ell-1]$, let $\widehat{y}_k \sim N(0,\rho^2 (k+2) /(2\ell))^n$. Clearly $\Pr[z \in S] = \Pr[x_0+\dots+x_{\ell-1} \in S]$. 
Let $f(z)=\mathbf{1}_{z \in S}$ represent the indicator function of $S$. For $0 \le k \le \ell-1$, let $$ E_k = \E_{\substack{\widehat{y}_k,z_1, \dots, z_k, x_{k+1}, \dots, x_{\ell-1}}}\left[ \prod_{ \zeta_1, \dots, \zeta_{k} \in \set{\pm 1}} f\Big(\widehat{y}_k+ \sum_{j=1}^k \zeta_j z_j + \sum_{j=k+1}^{\ell-1} x_j \Big)\right] $$

We will prove that for each $k \in [\ell-1]$, $E_{k-1}^2 \le E_k$. 
Using Cauchy-Schwartz inequality, we have
\begin{align*}
    E_{k-1}^2 &= \left(\E_{\substack{\widehat{y}_{k-1},z\\ x_{k+1}, \dots, x_{\ell-1}} }\E_{x_{k}} \Big[ \prod_{ \zeta_1, \dots, \zeta_{k-1} \in \set{\pm 1}} f\Big(\widehat{y}_{k-1}+ \textstyle\sum_{j=1}^{k-1} \zeta_j z_j + \sum_{j=k}^{\ell-1} x_j \Big) \Big]  \right)^2 \\
    & \le \E_{\substack{\widehat{y}_{k-1}, x_{k+1}, \dots, x_{\ell-1}\\ z_1, \dots, z_{k-1}}}\left(\E_{x_{k}} \Big[ \prod_{ \zeta_1, \dots, \zeta_{k-1} \in \set{\pm 1}} f\Big(\widehat{y}_{k-1}+ \textstyle\sum_{j=1}^{k-1} \zeta_j z_j + \sum_{j=k}^{\ell-1} x_j \Big) \Big]  \right)^2.
\end{align*}
Now if $y_k, z_k$ are i.i.d variables distributed as $N(0,\rho^2/(2\ell))^n$, then $x_k, y_k+z_k, y_k - z_k$ are identically distributed. More crucially, $y_k+z_k$ and $y_k - z_k$ are independent! Hence 
\begin{align*}
    E_{k-1}^2 &\le \E_{\substack{\widehat{y}_{k-1}, x_{k+1}, \dots, x_{\ell-1}\\ z_1, \dots, z_{k-1}}}\Big( \E_{y_{k}, z_k \sim N(0,\tfrac{\rho^2}{2\ell}))^n} \Big[ \prod_{ \zeta_1, \dots, \zeta_{k-1} \in \set{\pm 1}} f\Big(\widehat{y}_{k-1}+ \textstyle\sum_{j=1}^{k-1} \zeta_j z_j + (y_k + z_k)+ \sum_{j=k+1}^{\ell-1} x_j \Big) \Big] \\
    & ~~\times \E_{y_{k}, z_k \sim N(0,\tfrac{\rho^2}{2\ell}))^n} \Big[ \prod_{ \zeta_1, \dots, \zeta_{k-1} \in \set{\pm 1}} f\Big(\widehat{y}_{k-1}+ \textstyle\sum_{j=1}^{k-1} \zeta_j z_j + (y_k - z_k)+ \sum_{j=k+1}^{\ell-1} x_j \Big) \Big]  \Big)\\
    &=\E_{\substack{\widehat{y}_{k-1}, x_{k+1}, \dots, x_{\ell-1}\\ z_1, \dots, z_{k-1}}}\E_{y_{k}, z_k \sim N(0,\tfrac{\rho^2}{2\ell}))^n} \Big[ \prod_{ \zeta_1, \dots, \zeta_{k-1} \in \set{\pm 1}} f\Big(\widehat{y}_{k-1}+ \textstyle\sum_{j=1}^{k-1} \zeta_j z_j + (y_k + z_k) + \textstyle\sum_{j=k+1}^{\ell-1} x_j \Big) \\
    &~\qquad\qquad\qquad~~\times f\Big(\widehat{y}_{k-1}+ \textstyle\sum_{j=1}^{k-1} \zeta_j z_j + (y_k - z_k)+ \textstyle\sum_{j=k+1}^{\ell-1} x_j \Big) \Big] \\
    &= \E_{\substack{\widehat{y}_{k}, x_{k+1}, \dots, x_{\ell-1}\\ z_1, \dots, z_{k}}}\Big[ \prod_{ \zeta_1, \dots, \zeta_{k} \in \set{\pm 1}} f\Big(\widehat{y}_{k}+ \textstyle\sum_{j=1}^{k} \zeta_j z_j + \sum_{j=k+1}^{\ell-1} x_j \Big) \Big],
\end{align*}
where the last step follows by identifying $\widehat{y}_{k-1}+y_k$ with $\widehat{y}_k$. 
The proof of the lemma is completed by observing that $E_0=\Pr[\widehat{y}_0+x_1+\dots+x_\ell \in S]=\Pr[z \in S]$.

\end{proof}

We now proceed to the proof of the main decoupling statement. 
\begin{proof}[Proof of Proposition~\ref{prop:decoupling}]
Let $S:= \set{z \in \R^n: \norm{g(z+u) - t} \le \epsilon}$. Let $z_0 \sim N(0,\rho^2(\ell+1)/(2\ell))^n$ and $z_1, \dots, z_{\ell-1} \sim N(0,\rho^2/(2\ell))^n$ be independent $n$-variate Gaussian random variables. From Lemma~\ref{lem:positivecorrelation} we have for $z \sim N(0,\rho^2)^n$, 
\begin{align*}
    \Pr_{z}\Big[ \norm{g(z+u) -t} \le \epsilon  \Big]&\le \Pr_{z_0, \dots, z_{\ell-1}}\Big[ \bigwedge_{\zeta_1, \dots, \zeta_{\ell-1} \in \set{\pm 1}} \Big( \norm{g(u+z_0+ \textstyle\sum_{j=1}^{\ell-1} \zeta_j z_j)  -t} \le \epsilon \Big) \Big]^{1/(2^{\ell-1})}\\
   &\le \Pr_{z_0 \dots z_{\ell-1}}\Big[ \sum_{\zeta_1, \dots, \zeta_{\ell-1} \in \set{\pm 1}} \norm{ g(u+z_0+ \textstyle\sum_{j=1}^{\ell-1} \zeta_j z_j)  -t} \le 2^{\ell-1} \epsilon \Big]^{1/(2^{\ell-1})} \\
    &\le \Pr_{z_0 \dots z_{\ell-1}} 
    \Big[ \Big\| \sum_{\zeta_1, \dots, \zeta_{\ell-1} \in \set{\pm 1}} \big(\prod_{j=1}^{\ell-1} \zeta_{j}\big) g\big(u+z_0+ \textstyle\sum_{j=1}^{\ell-1} \zeta_j z_j\big) \Big\| \le 2^{\ell-1} \epsilon \Big]^{1/(2^{\ell-1})},
\end{align*}
\anote{4/2: Modified $g(z_0+ \textstyle\sum_{j=1}^{\ell-1} \zeta_j z_j)$ to $g(u+z_0+ \textstyle\sum_{j=1}^{\ell-1} \zeta_j z_j)$, and removed the $-t$.}
where the last inequality follows from triangle inequality, and observing that the signed combinations of $t$ cancel out when $\ell \ge 2$. Now applying Lemma~\ref{lem:deg-l} for each $i \in [m]$, we get
\begin{align*}
  \Pr_{z \sim N(0,\rho^2)^n}\Big[ \norm{g(z+u) -t} \le \epsilon  \Big]   & \le \Pr_{z_0, \dots, z_{\ell-1}}\Big[  \norm{ \widehat{g}(u+z_0, z_1, \dots, z_{\ell-1})} \le \epsilon/ \ell! \Big]^{1/(2^{\ell-1})}.
\end{align*}

\end{proof}

\subsection{Proofs of Theorem~\ref{thm:newdecoupling} and Theorem~\ref{thm:newvectorCW}}\label{sec:proof-decoupling}

\begin{proof}[Proof of Theorem~\ref{thm:newdecoupling}]
Let $m=\delta D$, and let $M_1, M_2, \dots, M_{m}$ be an orthonormal basis of symmetric tensors in $W \subset \R^{n^{\otimes \ell}}$. We will also denote by $M$ the $m \times n^{\ell}$ matrix formed by flattening $M_1, \dots, M_{m}$ respectively. For each $j \in [m]$, let 
$g_j(x)= \iprod{M_j, x^{\otimes \ell}}$. Let $\tilde{x}=x+z$ where $z \sim N(0,\rho^2/n)^n$. We would like to lower bound $\norm{\Pi_W \tilde{x}^{\otimes \ell}}_2= \norm{g(x+z)}_2$. Using Proposition~\ref{prop:decoupling} with $t=0$, for all $\epsilon>0$, we have 
\begin{align}
\Pr\Big[\norm{g(x+z)}_2 < \eps \Big]
    &\le \Pr\Big[ \norm{\Pi_W (x+z_0) \otimes z_1 \otimes z_2 \dots \otimes z_{\ell-1}}_2 < \eps/\ell! \Big]^{1/(2^{\ell-1})}, \label{eq:generalderivation}
\end{align}    
where $z_0 \sim N(0,\tfrac{\rho^2(\ell+1)}{2\ell n})^n$,
$z_1, z_2, \dots, z_{\ell-1} \sim N(0,\tfrac{\rho^2}{2\ell n})^n$. 
 Then
  \begin{align}
      \Pr\Big[\norm{\Pi_W \tilde{x}^{\otimes \ell}}_2 < \frac{c(\ell)\rho^\ell}{n^\ell} \Big] &  \le \exp\Big(-c'(\ell)\delta  n\Big),
\end{align}
with $c(\ell), c'(\ell)>0$ being constants that depend only on $\ell$. The last inequality follows from~(\ref{eq:generalderivation}) and Lemma~\ref{lem:quantitativebound} applied with $p=1/e,x_1=x, x_2=x_3=\dots=x_\ell=0$, and $\delta'=\delta/\ell^\ell$. This concludes the proof of Theorem~\ref{thm:newdecoupling}. 
\end{proof}

Please see Appendix~\ref{sec:combinatorialproof} for an alternate combinatorial proof when $\ell=2$. Note that we can also obtain a similar statement for general lower bound of $\eps \eta$ with $\eps \in (0,1/poly(n))$ (as in Theorem~\ref{thm:newvectorCW}), where the failure probability becomes $\eps^{\Omega_\ell(\delta n)}$. The proof is exactly the same, except that we can apply Lemma~\ref{lem:quantitativebound} with $p=\epsilon^{1/\ell}$ instead. Finally, the proof of Theorem~\ref{thm:newvectorCW} is almost identical to Theorem~\ref{thm:newdecoupling}. In fact Theorem~\ref{thm:newvectorCW} essentially corresponds to the special case of Theorem~\ref{thm:columnpolynomials} when $k=1$. We include a proof of Theorem~\ref{thm:newvectorCW} in Appendix~\ref{sec:newvectorCW}.
\anote{4/5: made changes.}

\input{polynomialcolumns.tex}

%% file: polynomialcolumns.tex
\subsection{Condition Number Lower Bounds for Arbitrary Polynomials} \label{sec:sym:finishingproof}

\newcommand{\projw}{\Pi_{\calW}}
\newcommand{\projperpw}{\Pi_{\calW^{\perp}}}

We are now ready to complete the proof of Theorem~\ref{thm:columnpolynomials}. We start by re-stating the theorem. 

\begin{theorem}[Same as Theorem~\ref{thm:columnpolynomials}]
Let $\ell \in \Z_+$ be a constant and let $a_1, a_2, \dots, a_k \in \R^{n}$ be any arbitrary collection of vectors, let $f_1, f_2, \dots, f_m$ be a collection of arbitrary homogeneous polynomials $f_i:\R^{n} \to \R $  of degree $\ell$ given by 
$$f_i(x)=\sum_{\substack{J=(j_1, \dots, j_\ell) \in {[n] \choose \ell} \\ j_1 \le j_2 \le \dots \le j_\ell}} U_i(j_1, \dots, j_\ell) x_{j_1} x_{j_2} \dots x_{j_\ell},$$
and let $M_f (a_1, \dots, a_k) = \big(f_i(a_j) \big)_{i \in [m], j\in [k]}$ be the $m \times k$ matrix formed by applying each of these polynomials with the $k$ vectors $a_1, \dots, a_k$. 
Denote by $U \in \R^{m \times D}$ with $D = {n+\ell-1 \choose \ell}$, with row $i \in [m]$ representing coefficients of $f_i$. 
We have that with probability at least $1-\exp\big(-\Omega_\ell(\delta n)+\log k \big)$ that
\begin{equation}\label{eq:columnpolynomials}
    \sigma_k\Big( M_f(\tilde{a}_1, \dots, \tilde{a}_k) \Big) \ge \frac{\Omega_\ell(1)}{\sqrt{k}}\cdot\frac{\rho^\ell}{n^\ell} \cdot \sigma_{k+\delta D}(U),
\end{equation}\cnote{4.4:write out the polynomial}
where $\tilde{a_j}$ represents a random perturbation of $a_j$ with independent Gaussian noise $N(0,\rho^2/n)^n$.
\end{theorem}
\begin{remark} We note that the condition on $U$ is almost tight, since $\sigma_k(U)$ being non-negligible is a necessary condition (irrespective of $A$). Proposition~\ref{claim:counterexample} shows that the additive $\delta n^{\ell}$ term in number of non-negligble singular values is necessary even when $k=1$. Also note that by choosing a projection matrix $U$ for a subspace of dimension $\delta D$, we recover Theorem~\ref{thm:newdecoupling}. Finally as before, we can obtain an analogous statement for $\eps \in (0,1/poly_\ell(n))$ as in Theorem~\ref{thm:newvectorCW} (see Section~\ref{sec:proof-decoupling}). 
\end{remark}

\begin{definition}
Let $D=\binom{n+\ell-1}{\ell}$. For $x_1,\cdots,x_n\in \mathbb{R}$, $P_\ell(x_1,\cdots,x_n)\in\mathbb{R}^D$ is a vector whose entries corresponding to D different degree-$\ell$ monomials of $x_1,\cdots,x_n$.
\end{definition}
The idea behind the proof is to view $M_f(a_1, \dots, a_k)$ as the product of a coefficient matrix and the matrix whose $i$th column is $P_\ell(a_i)$. Call the latter matrix $Y$. The following lemma show how to use the property that Theorem~\ref{thm:newvectorCW} gives about $Y$ to show Theorem~\ref{thm:columnpolynomials}.


\begin{lemma}\label{lem:leftmultiply}
Let $\delta \in (0,1)$, and let $U$ be a $D' \times D$ matrix, and let $Y \in \R^{D \times R}$ be a random matrix with \anote{11/28:added independent} independent columns $\tilde{Y}_1, \tilde{Y}_2, \dots, \tilde{Y}_R$ satisfying the following condition:
for each $j \in [R]$, and any fixed subspace $\calV$ of dimension at least $\delta D$, $\norm{\Pi_{\calV} \tilde{Y}_j}_2 \ge \kappa_1$ with probability at least $1-\gamma/R$ over the randomness in $\tilde{Y}_j$. Then assuming $\sigma_{R+\delta D}(U) \ge \kappa_2$, 
we have that $\sigma_R(UY) \ge \kappa_1 \kappa_2/\sqrt{R}$ with probability at least $1-\gamma$.
\end{lemma}
\begin{proof}
For convenience let $r:= R+\delta D$.
We will lower bound the minimum singular value of $M=UY$ using the leave-one-out-distance. Fix an $j \in [R]$; we want column $M_j=U\tilde{Y}_j$ to have a non-negligible component orthogonal to $\calW=span\big( \set{U\tilde{Y}_i: i \in [R], j \ne i }\big)$ w.h.p.

Let $\projw, \projperpw$ be the projectors onto the space $\calW, \calW^{\perp}$ respectively. Note that $\sigma_{r}(U)=\sigma_{R+\delta D} \ge \kappa_2$, and $\sigma_{D'-R+1}(\projperpw) \ge 1$. We can use the following robust version of Sylvester's inequality for products of matrices using the variational characterization of singular values to conclude 
\begin{align*}
    \sigma_{r-R+1}\big( \projperpw U \big) &\ge \sigma_{D'-R+1}\big( \projperpw \big) \sigma_r(U) \\
    &\ge \kappa_2.  
\end{align*}

Let $\calV$ be the subspace spanned by the top $r-R+1$ right singular vectors of $\projperpw U$. Since the dimension of $\calV$ is at least $r-R+1 \ge \delta D$, we can then use the condition of the lemma to conclude that with probability at least $1-\gamma/R$, $\norm{\Pi_\calV \tilde{Y}_j}_2 \ge \kappa_2 \kappa_1$. Hence, by using a union bound over all $j \in [R]$ and using the leave-one-out distance the lemma follows. 
\end{proof}
\anote{4/3: Removed the symmetrization lemma, and pushed it into the proof of Theorem below.}
We can now complete the proof of the main result of the section.
 \begin{proof}[Proof of Theorem~\ref{thm:columnpolynomials}]
 The idea is to apply  Lemma~\ref{lem:leftmultiply} with $D'=m, D = {n+\ell-1 \choose \ell}, R=k$, where $U$ is the corresponding coefficient matrix, and $Y$ is the matrix whose $j$th column is $\tilde{a}_j^{\otimes \ell}$. Note that the naive representation of $\tilde{a}_j^{\otimes \ell} \in \R^{n^{\otimes \ell}}$ is in $n^{\ell}$ dimensions, whereas the rows of the co-efficient matrix $U$ is in $\R^D$. However $\tilde{a}_j^{\otimes \ell}$ are elements of the $D$-dimensional space of symmetric tensors of order $\ell$ (alternately each row of $U$ can be seen as a $n^{\ell}$ dimensional vector constructed by flattening the corresponding symmetric order $\ell$ tensor for that row of $U$). Hence,
 Theorem~\ref{thm:newdecoupling} implies that $Y$ satisfies the conditions of Lemma~\ref{lem:leftmultiply}, and this completes the proof.
 \end{proof}


\subsection{Tight Example for Theorem~\ref{thm:columnpolynomials} and~\ref{thm:newvectorCW}}
We now give a simple example that demonstrates that the condition on many non-trivial singular values for the matrix $M$ that encodes $g$ is necessary. 

\begin{proposition}\label{claim:counterexample}
In the notation of Theorem~\ref{thm:newvectorCW}, for any $r \ge 1$, there exists a matrix $M \in \R^{m \times n^{\ell}}$ (where $m=rn^{\ell-1}$), with the $j$th row corresponding to a symmetric order $\ell$ tensor $M_j$, such that $\sigma_{r n^{\ell-1}}(M) = \Omega_\ell(1)$, but 
$$\Pr_{z \sim N(0,1/n)^n}\Big[\norm{g(z)}_2 = \norm{M z^{\otimes \ell}}\le \epsilon \Big] \ge (c\epsilon)^{O_\ell(r)},$$
for some absolute constant $c>0$.
\end{proposition}
Considering the subspace of symmetric tensors spanned by the rows of $M$ also gives a similar tight example for Theorem~\ref{thm:newdecoupling}. Moreover, the above example also gives a tight example for Theorem~\ref{thm:columnpolynomials} even when $k=1$, by considering the function $f(x):=g(x)$, and $a_1=0$ (so $\widetilde{a}_1=z$).

\anote{4/2: It doesn't show tightness of the minimum singular value (it loses an $n^{\ell/2}$ factor). Can this bound be made tight.. maybe Anari et al. or our use of it is not tight? }
\begin{proof}
Let $e_1, \dots, e_n$ constitute the standard basis for $\R^n$. Let $\calU$ be the space $\R^{n^{\ell-1}}$, and let $\calV \subset \R^n$ be the subspace spanned by $e_1,e_2,  \dots, e_r$.  Let $E_1, E_2, \dots, E_{n^{\ell-1}} \in \R^{n^{\ell-1}}$ constitute the standard basis of $\calU$ given by all the $\ell-1$ wise tensor products of $e_1, \dots, e_n$. Consider the product space $\calW=\calU \otimes \calV$, and let $B$ be the matrix whose $m=rn^{\ell-1}$ rows correspond to the orthonormal basis of $\calW$ given by $\set{E_I \otimes e_j:~ I \in [n]^{\ell-1},~ j \in [r]}$. Note that each of these vectors are $1$-sparse. Let $g:\R^n \to \R^m$ be given by $\forall j \in [m],~ g_j(x)=\iprod{B_j, x^{\otimes \ell} }$. First note that by definition, $\norm{g(x)}_2 = \norm{\Pi_{\calU \otimes \calV} x^{\otimes \ell}}_2$. Hence, if $z \sim N(0,1/n)^n$, we have
\begin{align*}
    \Pr_z\Big[ \norm{g(z)}_2 \le \epsilon \Big]&= \Pr_z\Big[ \norm{\Pi_{\calU \otimes \calV} z^{\otimes \ell}}_2 \le \epsilon \Big]\\
    &= \Pr\Big[ \norm{\Pi_{\calU} z^{\otimes (\ell-1)}}  \norm{\Pi_{\calV} z} \le \epsilon \Big] \\
    &\geq \Pr\Big[ \norm{\Pi_\calV z}\leq \epsilon/2,\norm{z}\leq 2^{1/(\ell-1)}\Big]\\
    &=\Pr\Big[  \norm{\Pi_\calV z}\leq \epsilon/2~|~\norm{z}\leq 2^{1/(\ell-1)}\Big]\cdot (1-o(1))\\
    &\geq\Pr\Big[  \norm{\Pi_\calV z}\leq \epsilon/2\Big]\cdot (1-o(1)) && \text{by Lemma~\ref{lem:deconditioning} in Appendix}\\
    &\ge (c\epsilon)^r, 
\end{align*}
for some absolute constant $c>0$, using standard properties of Gaussians. 

We now just need to give a lower bound of $\Omega(r n^{\ell-1})$ for the number of non-trivial singular values of the matrix $M$, where $M_j$ is the symmetric order $\ell$ tensor representing $g_j$ i.e.,  $\iprod{M_j, x^{\otimes \ell}}=\iprod{B_j, x^{\otimes \ell}}$ for every $x \in \R^n$. In other words $M_j$ is just the symmetrization (projection onto the space of all symmetric tensors) of $B_j$. Note that each $M_j$ is $\ell!$ sparse (since $B_j$ were $1$-sparse). Hence there are at least $rn^{\ell-1}/ \ell!$ vectors $M_j$ which have disjoint support. Hence at least $r n^{\ell-1}/ \ell!$ singular values of $M$ are at least $1/\sqrt{\ell!}$, as required.   
\end{proof}

%% file: polynomials.tex


\anote{11/28: Edited the first sentence.}
In this section, we consider random matrix ensembles, where each column is a constant degree ``monomial'' involving a few of the columns.  
We will first consider a matrix $M$ whose columns are degree $\ell$ monomials in the input vectors $\ta_1, \ldots, \ta_k$ (that is, tensors of the form $\ta_{f(1)} \otimes \ldots \otimes \ta_{f(\ell)}$ with $f(i) \in [k]$ for $i = 1, \ldots, \ell$). Since the same vector may appear in many columns or multiple times within the same column, there are now dependencies in the perturbations between columns as well as within a column, so we cannot apply \cite{BCMV} directly. We deal with these dependencies by extending an idea of Ma, Shi and Steurer~\cite{MSS}, carefully defining appropriate subspaces that will allow us to decouple the randomness. 

Since one type of dependence comes from the same input vector appearing in many different columns, it is natural to require that the number of these overlaps be small. Because of the decoupling technique used to avoid dependencies within a column, the troublesome overlaps are only those in which the same input vector appears in two different columns of $M$ in the same position within the tensor product. This motivates the following definition.
\begin{definition}
Let $M$ be a matrix whose columns $M_1,\ldots, M_R$ consist of order-$\ell$ tensor products of $\{\ta_1, \ldots, \ta_k\}$. For $s \in [\ell]$ and a fixed column $M_i$, let $\Delta_s(i)$ be the number of other columns that differ from $M_i$ in exactly $s$ spots. (If $M_i = \ta_{f(1)} \otimes \ldots \otimes \ta_{f(\ell)}$ and $M_j = \ta_{f'(1)} \otimes \ldots \otimes \ta_{f'(\ell)}$, then the number of spots in which $M_i$ and $M_j$ differ is $|\{i : f(i) \neq f'(i)\}|$.) Finally, let $\Delta_s = \max_i \Delta_s(i)$.
\end{definition}

\begin{theorem}[Same as Theorem~\ref{thm:indepmonomials}]
\label{thm:indepmonomials2}
Let $\{\ta_1, \ldots, \ta_k\} \subseteq \R^n$ be a set of $\rho$-perturbed vectors, let $\ell \in \Z_+$ be a constant, and let $M \in \R^{n^\ell \times R}$ be a matrix whose columns $M_1, \ldots, M_R$ are tensor monomials in $\{\ta_i\}$. Let $\Delta_s$ be as above for $s = 1, \ldots, \ell$. If
\begin{equation}\label{eq:indepmonomials}
    \sum_{s=1}^\ell \Delta_s \cdot \left(\frac{n}{\ell}\right)^{\ell-s} \le c\left(\frac{n}{\ell}\right)^\ell
\end{equation}
for some $c \in (0, 1)$, then $\sigma_R(M) > \Omega_\ell(1) \cdot (\rho/n)^\ell/\sqrt{R}$ with probability at least $1 - \exp(-\Omega_\ell(1) (1-c) n + \log R)$.
\end{theorem}

\begin{remark}
The condition~\eqref{eq:indepmonomials} is tight up to a multiplicative constant depending only on $\ell$. We give a simple upper bound on $\Delta_s$. Assume $\sigma_R(M) >0$, and fix a column $M_i$ of $M$. 
There are $\binom{\ell}{s}$ ways to choose a set of $s$ spots in which to differ from $M_i$, and once we make this choice, the dimension of the available space is $n^s$ since each of the $s$ spots contributes $n$ dimensions. Therefore the subspace of $\R^{n^\ell}$ consisting of all tensors that differ from $M_i$ in exactly $s$ spots has dimension at most $\binom{\ell}{s}n^s$. Since all subsets of columns of $M$ must be linearly independent, we must have $\Delta_s \le \binom{\ell}{s}n^s$. Therefore our condition is tight up to a factor of at most $\ell^{2\ell+1}$.
\end{remark}

In the above theorem, as stated, the columns of $M$ are ``monomials'' involving the underlying vectors $\ta_1,\dots, \ta_k$. However in our applications (e.g., Sections~\ref{sec:HMM} and \ref{sec:foobi}) the matrix of interest $M'$ will have columns that are more general polynomials of the underlying vectors. Such matrices are expressible as $M'=MP$ where $P \in \R^{R \times R'}$ is a coefficient matrix with $\sigma_{R'}(P) >1/\poly(n,1/\rho)$. Hence, our theorem implies that $\sigma_{R'}(M') > 1/\poly(n,1/\rho)$ in these cases w.h.p.  

As in \cite{BCMV}, we will use leave-one-out distance, denoted $\ell(M)$, as a surrogate for the smallest singular value.
The proof will make use of Lemma~\ref{lem:quantitativebound}, which we will use to bound leave-one-out distances. Our goal will be to find a suitable subspace $W$ that is both large enough and independent of the column of $M$ we are projecting. 

\begin{proof}[Proof of Theorem~\ref{thm:indepmonomials2}]
Let $L_1, \ldots, L_\ell$ be an equipartition of $[n]$. Define a new matrix $M' \in \R^{(\frac{n}{\ell})^\ell \times R}$ by restricting the columns of $M$ to the indices $L_1 \times L_2 \times \ldots \times L_\ell$. In other words, if $M_i$ is a column of $M$ with $M_i = \ta_{f(1)}\otimes \ldots \otimes \ta_{f(\ell)}$, then $M_i'= \ta_{f(1),L_1} \otimes \ldots \otimes \ta_{f(\ell),L_\ell}$, where $a_L$ denotes the restriction of the vector $a$ to the coordinates in the set $L$. This ensures that for every column $M_i$, the perturbations of each factor of this tensor product are independent.

Fix a column $M_i'$ of $M'$, and let $W$ be the subspace spanned by all other columns of $M'$. We want to find a subspace $V$ satisfying:
\begin{enumerate}
    \item $W \subseteq V$.
    \item $V$ is independent of $M_i'$.
    \item $\dim V^\perp = c'(\frac{n}{\ell})^\ell$ for some $c' \in (0,1).$
\end{enumerate}
Given such a $V$, properties 2 and 3 allow us to apply Lemma~\ref{lem:quantitativebound} to obtain that $\|\Proj_{V^\perp} M_i'\| \ge \Omega_\ell((\rho/n)^\ell)$ with probability at least $1 - \exp(-\Omega(c'n))$. Since $W \subseteq V$, we have 
$$\|\Proj_{W^\perp} M_i'\| \ge \|\Proj_{V^\perp} M_i'\| \ge \Omega_\ell((\rho/n)^\ell)$$
with high probability. Taking a union bound over all columns of $M'$ gives that $\ell(M') \ge \Omega_\ell((\rho/n)^\ell)$ with probability at least $1 - \exp(-\Omega_\ell(1)\cdot c'n + \log R)$. Since adding more rows to $M'$ can only increase the magnitude of the projection of any column onto some subspace, $\ell(M) \ge \ell(M')$. Now using properties of the leave-one-out distance (Lemma~\ref{lem:leave-one-out}), we have 
$$\sigma_{min}(M) \ge \frac{\ell(M)}{\sqrt{R}} \ge \Omega_\ell(1)\cdot\frac{\rho^\ell}{n^\ell\sqrt{R}}.$$


Next we construct the subspace $V$. Let $M_{i'}'$, $i' \neq i$ be some other column of $M'$. Let $S \subseteq [\ell]$ be the set of indices at which $M_i'$ and $M_{i'}'$ share a factor, and let $s = |S|$. In order to ensure $V$ is independent of $M_i'$, we must avoid touching any factors of $M_{i'}'$ shared by $M_i'$. Therefore we include in $V$ all vectors of the form $\tu_1 \otimes \ldots \otimes \tu_\ell$, where $\tu_j$ agrees with the $j$th factor of $M_{i'}'$ if $j \not\in S$ and $\tu_j$ is any vector in $\R^{n/\ell}$ otherwise. As desired, $V$ now includes $M_{i'}'$ and is independent of $M_i'$, at a cost of adding $(\frac{n}{\ell})^s$ dimensions to $V$. 

Repeat this process for each $i' \neq i$, and let $V$ be the span of all vectors included at each step. Since the number of overlaps with $M_i'$ can be $s$ at most $\Delta_{\ell-s}$ times, the total dimension of $V$ is at most $\sum_{s=1}^\ell \Delta_s(\frac{n}{\ell})^{\ell-s}$. By our assumption on the $\Delta_s$s, we get $\dim V^\perp = c'(\frac{n}{\ell})^\ell$ as desired, with $c' = 1-c$.
\end{proof}


%% file: subspacerecovery.tex

\anote{4/3: We need to add something about failure probability, and error tolerance as a remark. Also how the tensoring algorithm is new...}
We introduce the following smoothed analysis framework for studying robust subspace recovery. The following model also tolerates some small amount of error in each point i.e., inliers need not lie exactly on the subspace, but just close to it.

\subsection{Input model}\label{sec:subspace-rec-model}

In what follows, $\alpha, \eps_0, \rho \in (0,1)$ are parameters.

\begin{enumerate}
\item An adversary chooses a hidden subspace $T$ of dimension $d$
in $\mathbb{R}^n$, and then chooses $\alpha m$ points from $T$ and
$(1-\alpha)m$ points from $\mathbb{R}^n$.  We denote these points
inliers and outliers respectively. Then the adversary mixes them in
arbitrary order. Denote these points $a_1, a_2, \dots, a_m$. Let
$A=(a_1, a_2, \dots, a_m)$, and $I_{in}, I_{out}$ be the set of
indices of inliers and outliers respectively. For convenience, we
assume that all the points have lengths in the range $[1/2,
1]$.\footnote{If the perturbations in step (2) are done proportional
to the norm, this assumption can be made without loss of generality.
(Since the algorithm can scale the lengths of each of the points.)}
\item Each inlier is $\rho$-perturbed with respect to $T$.
(Formally, this means considering an orthonormal basis $B_T$ for $T$
and adding $B_T v$, where $v \sim \calN(0, \rho^2/d)^d$.) Each outlier
is $\rho$-perturbed with respect to $\mathbb{R}^n$. Let $G$ denote the
perturbations, and let us write $\widetilde{A}=A+G$.
\item With the constraint $\norm{E}_F \leq \epsilon_0$, the
adversary adds noise $E\in \mathbb{R}^{n\times m}$ to $A$, yielding
$\widetilde{A}'=\widetilde{A}+E =
(\widetilde{a}_1',\widetilde{a}_2',\cdots)$. Note that this adversarial noise can depend on the random perturbations in step 2.
\item We are given $\widetilde{A}'$.
\end{enumerate}

The goal in the subspace recovery problem is to return a subspace $T'$ close to $T$.

\paragraph{Notation.}  As introduced above, $\tilde{A} = A+G$ denotes the perturbed vectors. $\tilde{a}_i$ denotes the $i$'th column of $\tA$. We also use the notation $A_I$ to denote the sub-matrix of $A$ corresponding to columns in a set $I$.

\subsection{Our result}

We show the following theorem about the recoverability of $T$.

\begin{theorem}\label{thm:robustrecovery:maintheorem}
\cnote{should be $m\geq n^\ell+8d/(\delta\alpha)$} \anote{11/25: Check
the lower bound on $m$.}

Let $\delta\in(0,1)$, $\ell \in \mathbb{Z}_+$ and $\rho>0$. Suppose we are given $m \ge n^\ell+8d/(\delta \alpha)$ points $x_1, x_2,\cdots, x_m\in \mathbb{R}^n$ generated as described above, where the fraction of inliers $\alpha$ satisfies $\alpha \geq (1+\delta) \binom{d+\ell-1}{\ell} / \binom{n+\ell-1}{\ell}$. Then there exists $\epsilon_0 = \poly_{\ell}(\rho/m)$\cnote{3.31:delete dependency on $\delta$} such that whenever $\norm{E}_F \le \eps_0$, there is an efficient deterministic
algorithm that returns a subspace $T'$ that satisfies
\begin{equation}
\norm{sin\Theta(T,T')}_F\leq\norm{E}_F\cdot{\poly_\ell(m,1/\rho)} \text {, w.p. } \ge 1-2m^2 [\exp(-\Omega_\ell(\delta n)) +\exp(-\Omega(d\log m))].
\end{equation}
\bnote{Change the error probability using Anari et al.}\cnote{3.31: done}\cnote{4.3: the guarantee should be $sin\Theta\leq \norm{}\cdot poly$}
\end{theorem}


When $d/n < 1$, the above theorem gives recovery guarantees even when the fraction of inliers is approximately $(d/n)^\ell$. This can be significantly smaller than $d/n$ (shown in~\cite{HM13}) for any constant $\ell > 1$.

\paragraph{Algorithm overview.}  We start by recalling the approach of~\cite{HM13}. The main insight there is that if we sample a set of size slightly less than $n$ from the input, and if the fraction of inliers is $> (1+\delta) d/n$, then there is a good probability of obtaining $>d$ inliers, and thus there exist points that are in the linear span of the others. Further, since we sampled fewer than $n$ points and the outliers are also in general position, one can conclude that the only points that are in the linear span of the other points are the inliers! In our algorithm, the key idea is to use the same overall structure, but with tensored vectors. Let us illustrate in the case $\ell = 2$. Suppose that the fraction of inliers is $> (1+\delta) \binom{d+1}{2} / \binom{n+1}{2}$. Suppose we take a sample of size slightly less than  $\binom{n+1}{2}$ points from the input, and consider the flattened vectors $x \ot x$ of these points. As long as we have more than $\binom{d+1}{2}$ inliers, we expect to find linear dependencies among the tensored inlier vectors. Further, using  Theorem~\ref{thm:newdecoupling} (with some modifications, as we will discuss), we can show that such dependencies cannot involve the outliers. 
This allows us to find sufficiently many inliers, which in turn allows us to recover the subspace $T$ up to a small error.


Given $m$ points, the algorithm (Algorithm~\ref{alg:robust}) considers several batches of points each of size $b = (1- \frac{\delta}{3}) {n+\ell-1 \choose \ell}$. Suppose for now that $m$ is a multiple of $b$, and that the $m/b$ batches form an arbitrary partition of the $m$ points. (See the note in Section~\ref{sec:non-divisible} for handling the general case.) In every batch, the algorithm does the following: for each point $u$ in the batch, it attempts to represent $u^{\ot \ell}$ as a ``small-coefficient'' linear combination (defined formally below) of the tensor products of the other points in the batch. If the error in this representation is small enough, the point is identified as an inlier. 

\begin{definition}[$c$-bounded linear combination]
Let $v_1, v_2, \dots, v_m$ be a set of vectors. A vector $u$ is said to be expressible as a $c$-bounded linear combination of the $\{v_i\}$ if there exist $\{ \alpha_i\}_{i=1}^m$ such that $|\alpha_i| \le c$ for all i, and $u = \sum_i \alpha_i v_i$. Further, $u$ is said to be expressible as a $c$-bounded combination of the $\{v_i\}$ with error $\delta$ if there exist $\{\alpha_i\}_{i=1}^m$ as above with $|\alpha_i|\le c$ for all $i$, and $\norm{u - \sum_i \alpha_i v_i}_1 \le \delta$.
\end{definition}
\cnote{4.3: the error is now measured in $\ell_1$ norm}
Notice that in the above definition, the error is measured by $\ell_1$ norm. In the algorithm, we will need a subprocedure to check whether a vector is expressible as a 1-bounded combination of some other vectors with some fixed error. By the choice of $\ell_1$ norm, this subprocedure can be formulated as a Linear Programming problem, hence we can solve it efficiently. 
\begin{algorithm}
\caption{Robust subspace recovery}\label{alg:robust}
\begin{algorithmic}[1]
\State Set threshold $\tau=\Omega_\ell(\rho^\ell/n^\ell)$(which is the threshold from Theorem~\ref{thm:newdecoupling}). Set batchsize $b = (1-\nicefrac{\delta}{3})\binom{n+\ell-1}{\ell}$.
\State Let $V_1, V_2, \cdots, V_{r}$ be the $r \le m$ batches each of size $b$ as defined above. 
\State Initialize $C=\emptyset$.
\For{$i=1, 2, \cdots, r$}
    \State\label{alg:step:S} Let $S$ be the set of all $u \in V_i$ such that $\tilde{a}_u'^{\otimes \ell}$ can be expressed as $1$-bounded combinations of  $\{\tilde{a}_v'^{\otimes \ell}: v \in V_i \setminus \{u\} \}$, with error $\leq \tau/2.$\cnote{3.31:change to $\tau/2$} 
    \State $C=C\cup S$
\EndFor
\State Return the subspace $T'$ corresponding to the top $d$ singular values of $\tilde{A}'_{\overline{C}}$, for any $2d$-sized subset $\overline{C}$ of $C$ \label{robust:line:end}
\end{algorithmic}
\end{algorithm}


\cnote{3.31: 1.modify threshold according new section 3.\\
2.modify step~\ref{alg:step:S} since S is index set}
\paragraph{Proof outline.}  The analysis involves two key steps. The first is to prove that none of the outliers are included in $S$ in step~\ref{alg:step:S} of the algorithm. This is where we use $1$-bounded linear combinations. If the coefficients were to be unrestricted, then because the error matrix $E$ is arbitrary, it is possible to have a tensored outlier being expressible as a linear combination of the other tensored vectors in the batch. 
The second step is to prove that we find enough inliers overall. 
On average, we expect to find at least $\tfrac{\delta}{3} \binom{d+\ell-1}{\ell}$ inlier columns in each batch. We ``collect'' these inliers until we get a total of $2d$ inliers. Finally, we prove that these can be used to obtain $T$ up to a small error.

%
%
%



For convenience, let us write $g(n) := \Omega_\ell (\delta n)$ (which is the exponent in the failure probability from Theorem~\ref{thm:newdecoupling}).  Thus the failure probabilities can be written as $\exp(-g(n))$. \bnote{Change the prob values to new ones..}\cnote{3.31: done}

\begin{lemma}\label{robust:lemma:nooutlier}
With probability at least $1- \exp(-g(n) + 2\log m)$, none of the outliers are chosen. I.e., $C\cap I_{out}=\emptyset$. \end{lemma}

\newcommand{\su}[1]{^{(#1)}}
\newcommand{\dist}{\text{dist}}
\newcommand{\spn}{\text{span}}

\begin{proof}

The proof relies crucially on the choice of the batch size. Let us fix some batch $V_j$.  Note that by the way the points are generated, each point in $V_j$ is $\widetilde{a_i}'$, for some $a_i$ that is either an inlier or an outlier.

Let us first consider only the perturbations (i.e., without the noise addition step). Recall that we denoted these vectors by $\widetilde{a}_i$. Let us additionally denote by $B\su{j}$ the matrix whose columns are $\widetilde{a}_i^{\ot \ell}$ for all $i$ in the phase $j$. Consider any $i$ corresponding to an outlier. Now, because the batch size is only $(1-\frac{\delta}{3}) \binom{n+\ell-1}{\ell}$, we have (using Theorem~\ref{thm:newdecoupling}) that the projection of the column $B\su{j}_i$ orthogonal to the span of the remaining columns (which we denote by $B\su{j}_{- i}$) is large enough, with very high probability. Formally,

\begin{equation}\label{eq:orth-proj-outlier}
\Pr[ \dist (B\su{j}_i, \spn (B\su{j}_{-i})) \ge \tau ] \ge 1- \exp(-g(n)).
\end{equation}

Indeed, taking a union bound, we have that the inequality $\dist(B\su{j}_i, \spn(B\su{j}_{-i})) \ge \tau$ holds for all outliers $i$ (and their corresponding batch $j$) with probability $\ge 1 - m^2 \exp(-g(n))$.\cnote{3.31: should be $m^2$}

We need to show that moving from the vectors $\widetilde{a}_i$ to $\widetilde{a}_i'$ maintains the distance. For this, the following simple observation will be useful.

\begin{observation}\label{obs:length-perturb}
If $a_i$ is an outlier, then 
\[ \Pr[ \norm{\widetilde{a}_i} \ge 1+ 2\rho ] \le \exp(-n/2). \]
On the other hand if $a_i$ is an inlier, 
\[ \Pr[ \norm{\widetilde{a}_i} \ge 1+ 4\rho \sqrt{\log m} ] \le \exp(-4d \log m). \]
\end{observation}
Both the inequalities are simple consequences of the fact that the vectors $a_i$ were unit length to start with, and are perturbed by $\mathcal{N}(0, \rho^2/n)$ and $\mathcal{N}(0, \rho^2/d)$ respectively. 

Now let us consider the vectors with noise added, $\widetilde{a}_i'$. Note that $\norm{\widetilde{a}_i - \widetilde{a}_i'} \le \eps_0$. Since $\norm{a_i} \le 1$ and since $i$ is an outlier, we have (using Observation~\ref{obs:length-perturb}), $\norm{\widetilde{a}_i'} \le 1+2\rho + \eps_0$, with probability $\ge 1- \exp(-n/2)$.
Thus for the flattened vectors $\widetilde{a}_i^{\ot \ell}$, with the same probability,
\begin{align}
\norm{\widetilde{a}_i^{\ot \ell} - (\widetilde{a}_i')^{\ot \ell}} &= \norm{\left( \widetilde{a}_i^{\ot \ell} - \widetilde{a}_i^{\ot (\ell-1)} \ot \widetilde{a}_i'\right) + \left( \widetilde{a}_i^{\ot (\ell-1)} \ot \widetilde{a}_i' - \widetilde{a}_i^{\ot (\ell-2)} \ot \widetilde{a}_i'^{\ot 2} \right) + \dots} \notag\\
&\le \ell (\max\{ \norm{\widetilde{a}_i}, \norm{\widetilde{a}_i'} \})^{\ell-1} \eps_0 \notag\\
&\le \ell (1+2\rho + \eps_0)^{\ell} \eps_0.
\label{eq:norm-diff-power}
\end{align}

Thus, for any $1$-bounded linear combination of the $b$ vectors in the batch (which may contain both inliers and outliers), $\widetilde{a}_i'^{\ot \ell}$ is at a Euclidean distance $\le b \ell (1+\eps_0 + 4\rho \sqrt{\log m})^{\ell} \eps_0$ to the corresponding linear combination of the $\ell$th powers of the vectors in the batch prior to the addition of noise (i.e., the columns of $B\su{j}_{-i}$). Thus if $b \ell (1+\eps_0 + 4\rho\sqrt{\log m})^{\ell} \eps_0 < \tau/2$, then $\widetilde{a}_i'^{\ot \ell}$ cannot be expressed as a $1$-bounded combination of the other lifted vectors in the batch with Euclidean error $< \tau/2$, let alone $\ell_1$ error.

This means that none of the outliers are added to the set $S$, with probability at least $1 - m [ \exp(-g(n)) - \exp(-4d \log m)]$.
\end{proof}

Next, we turn to proving that sufficiently many inliers are added to $S$. The following simple lemma will help us show that restricting to $1$-bounded combinations does not hurt us.

\begin{lemma}\label{lem:small-combinations}
Let $u_1, u_2, \dots, u_{d+c}$ be vectors that all lie in a $d$-dimensional subspace of $\R^n$. Then at least $c$ of the $u_i$ can be expressed as $1$-bounded linear combinations of $\{u_j\}_{j \ne i}$.
\end{lemma}
\begin{proof}
As the vectors lie in a $d$-dimensional subspace, there exists a non-zero linear combination of the vectors that adds up to zero. Suppose $\sum_i \alpha_i u_i = 0$. Choose the $i$ with the largest value of $|\alpha_i|$. This $u_i$ can clearly be expressed as a $1$-bounded linear combination of $\{u_j\}_{j \ne i}$. 

Now, remove the $u_i$ from the set of vectors. We are left with $d+c-1$ vectors, and we can use the same argument inductively to show that we can find $c-1$ other vectors with the desired property. This completes the proof.
\end{proof}




The next lemma now proves that the set $C$ at the end of the algorithm is large enough.
\begin{lemma}
\label{robust:lemma:alotinlier}
For the values of the parameters chosen above, we have that at the end of the algorithm,
\[ |C|\geq \frac{\delta/3}{1+\delta/3}\alpha m, \text{ with probability at least $1- \exp(-4 d \log m)$}. \]
\end{lemma}
\bnote{Is this probability high enough?}\cnote{I think so.}
\begin{proof}
We start with the following corollary to Lemma~\ref{lem:small-combinations}.  Let us consider the $j$th batch. 

\vspace{4pt}
\noindent {\em Observation.} Let $n_j$ be the number of inliers in the $j$th batch. If $n_j \ge \binom{d+\ell-1}{\ell}+k$, then the size of $S$ found in Step~\ref{alg:step:S} of the algorithm is at least $k$. 
\begin{proof}[Proof of Observation]
Define $B\su{j}$ as in the proof of Lemma~\ref{robust:lemma:nooutlier}. Now, since the inliers are all perturbed within the target subspace, we have that the vectors $\widetilde{a}_i^{\ot \ell}$ corresponding to the inliers all live in a space of dimension $\binom{d+\ell-1}{\ell}$.  Thus by Lemma~\ref{lem:small-combinations}, at least $k$ of the vectors $B\su{j}_i$ can be written as $1$-bounded linear combinations of the vectors $B\su{j}_{-i}$. 

For inliers $i$, using the fact that $a_i$ are perturbed by $\calN(0, \rho^2/d)$, we have
\[ \Pr[ \norm{\widetilde{a}_i} \geq (1+4\rho \sqrt{\log m}) ] \le \exp( - 4d \log m).\]
Using~\eqref{eq:norm-diff-power} again, we have that $\widetilde{a}_i'^{\ot \ell}$ can be expressed as a $1$-bounded linear combination of the other vectors in the batch, with  Euclidean error bounded by $b \ell \cdot (1+5\rho \sqrt{\log m})^{\ell} \eps_0$. We know $\ell_1$ norm is a $\sqrt{n^\ell}$-approximation of $\ell_2$ norm. By assumption, the $\ell_1$ norm of the error is $< \tau/2$, thereby completing the proof of the observation.
\end{proof}

Now, note that we have $\sum_j n_j \ge \alpha m$, by assumption.  This implies that 
\[ \sum_{j=1}^{m/b} \max \left\{ 0, n_j - \binom{d+\ell-1}{\ell} \right\} \ge \alpha m - \frac{m}{b} \binom{d+\ell-1}{\ell} \ge \frac{\delta/3}{1+\delta/3}  \alpha m.\]
The last inequality follows from our choice of $\alpha$ and the batch size $b$. Thus the size of $S$ in the end satisfies the desired lower bound.
\end{proof}

Finally, we prove that using any set of $2d$ inliers, we can obtain a good enough approximation of the space $T$, with high probability (over the choice of the perturbations). The probability will be high enough that we can take a union bound over all $2d$-sized subsets of $[m]$.  
\anote{11/25: shouldn't the bound of $\norm{E}_F \le poly(\rho/n)$ suffice?}
\begin{lemma}\label{robust:lemma:robustsvd}
Let $I \subseteq I_{in}$ be any (fixed) set of size $2d$. Then if $\norm{E}_F\leq \poly(\rho/m)$, the subspace $U$ corresponding to the top $d$ singular value of $\tilde{A}'_I$ will satisfy
\[ \norm{\sin \Theta(U,T)}_F\leq{\poly(m,1/\rho) }\cdot \norm{E}_F\] 
with probability at least $1- e^{-4d\log{m}}$.
\end{lemma}
\begin{proof}
We start by considering the matrix $\tilde{A}_I$ (the matrix without addition of error).  This matrix has rank $\le d$ (as all the columns lie in the subspace $T$). The first step is to argue that $\sigma_d (\tilde{A}_I)$ is large enough. This implies that the space of the top $d$ SVD directions is precisely $T$.  Then by using Wedin's theorem~\cite{Wed72}, the top $d$ SVD space $U$ of $\tilde{A}_I'$ satisfies
\begin{align}
\norm{\sin \Theta(U,T)}_F\leq \frac{2\sqrt{d}\norm{E}_F}{\sigma_d(\tilde{A_I})-\norm{E}_F}.
\end{align}
Hence it suffices to show $\sigma_d(\tilde{A})$ is at least inverse-polynomial with high probability. 


Recall that $\tA_I = A_I + G_I$, where $G_I$ is a random matrix. \anote{11/25: Added} Without loss of generality we can assume that $T$ is spanned by the first $d$ co-ordinate basis; in this case every non-zero entry of $G_I$ is independently sampled from $\mathcal{N}(0,\frac{\rho^2}{d})$. We can thus regard $A_I, G_I$ as being $d \times 2d$ matrices. 
 Recall that leave-one-out distance is a good approximation of least singular value, it suffices to show $\ell((A_I+G_I)^T)$ is at least inverse-polynomial with high probability. Let $A_j,G_j$ denote the $j$th row of $A_I,G_I$ correspondingly. Consider $j\in[d]$, fix all other rows except $j$th. Let $W$ be the subspace of $\R^{2d}$ that is orthogonal to $\text{span}(\{A_k+G_k:k\in[d],k\neq j\})$, \cnote{4.4: shall we say something here like: let us assume dim(W)=d+1, since dim(W)>d+1 can only help us. ?}and let $w_1, w_2, \dots, w_{d+1}$ be an orthonormal basis for $W$. Then for any $t>0$, if the projection of $(A_j + G_j)$ to $W$ is $<t$ (equivalent to the leave-one-out distance $<t$), then for all $1 \le i \le d+1$, we must have $|\iprod{w_i, A_j + G_j}| \le t$. Using the anti-concentration of a Gaussian and the orthogonality of the $w_i$, this probability can be bounded by $(t/\rho)^{d+1}$. Choosing $t = \rho/m^4$, this can be made $< (1/m^4)^{d+1}$, and thus after taking a union bound over the $m$ choices of $j$, we have that the leave-one-out distance is $> \rho/m^4$ (and thus $\sigma_d > \rho/m^5$) with probability $\ge 1 - \exp(-4d \log m)$
\end{proof}
We can now complete the proof of the theorem.
\begin{proof}[Proof of Theorem~\ref{thm:robustrecovery:maintheorem}]

Suppose that $\norm{E}_F\leq\epsilon_0$ is small enough. Now by Lemma~\ref{robust:lemma:nooutlier}, we have that $C\subseteq I_{in}$ with probability at least $1- \exp(-g(n) +\log m$.
%
By Lemma~\ref{robust:lemma:alotinlier} and our assumption that $m$ is at least $\Omega(d/(\delta \alpha))$, we know $|C|\geq 2d$ with probability $1-e^{-4d\log m}$. Finally, by Lemma~\ref{robust:lemma:robustsvd} and a union bound over all $2d$ sized subsets of $[m]$, we have that with probability at least $1 - \exp(-\Omega(d \log m))$, 
for any subset of inliers with size $2d$, the subspace $T'$ corresponding to the top-$d$ singular value will satisfy $\norm{\sin\Theta(T,T')}_F\leq\norm{E}_F/\poly(m)$.
\end{proof}

\subsection{Batches when $m$ is not a multiple of $b$}\label{sec:non-divisible}
he case of $m$ not being a multiple of $b$ needs some care because we cannot simply ignore say the last few points (most of the inliers may be in that portion). But we can handle it as follows: let $m'$ be the largest multiple of $b$ that is $<m$. Clearly $m' > m/2$. Now for $1 \le j \le n$, define $\calD_j = \{x_{j}, x_{j+1}, \dots, x_{j+m'-1}\}$ (with the understanding that $x_{n+t} = x_t$). This is a set of $m'$ points for every choice of $j$. Each $\calD_j$ is a possible input to the algorithm, and it has at least $m' > m/2$ points, and additionally the property that $b|m'$.

At least one of the $\calD_j$ has $\ge \alpha$ fraction of its points being inliers (by averaging). Thus the procedure above (and the guarantees) can be applied to recover the space. To ensure that no outlier is chosen in step~\ref{alg:step:S} of the algorithm (Lemma~\ref{robust:lemma:nooutlier}), we take an additional union bound to ensure that Lemma~\ref{robust:lemma:nooutlier} holds for all $\calD_j$.

%% file: HMM.tex


We consider the setup of Hidden Markov Models considered in \cite{AMR09, AHK12}. 
A hidden state sequence $Z_1, Z_2, \ldots, Z_m \in [r]$ forms a stationary Markov chain with transition matrix $P$ and initial distribution $w = \{w_k\}_{k \in [r]}$, assumed to be the stationary distribution. The observations $\{X_t\}_{t \in [m]}$ are vectors in $\R^n$.  The observation matrix of the HMM is denoted by $\mathcal{O} \in \R^{n\times r}$; the columns of $\mathcal{O}$ represent the conditional means of the observation $X_t \in \R^n$ conditioned on the hidden state $Z_t$ i.e., $\E[X_t \vert Z_t=i] = \mathcal{O}_i$, where $\mathcal{O}_i$ represents the $i$th column of $\mathcal{O}$. We also assume that $X_t$ has a subgaussian distribution about its mean (e.g., $X_t$ is distributed as a multivariate Gaussian with mean $\mathcal{O}_i$ when the hidden state $Z_t=i$).
\anote{11/28: edited above.}
In the smoothed analysis setting, the model is generated using a randomly perturbed observation matrix $\tOcal$, obtained by adding independent Gaussian random vectors drawn from  $N(0,\rho^2/n)^n$ to each column of $\mathcal{O}$. We remark that some prior works \cite{AMR09, Sharanetal} consider the more restrictive discrete setting where the observations are discrete over an alphabet of size $n$.\footnote{These observations can be represented using the $n$ standard basis vectors for the $n$ alphabets and column $\mathcal{O}_i$ gives the probability distribution conditioned on the current state being $i \in [r]$.} While our smoothed analysis model with small Gaussian perturbations is natural for the more general continuous setting, it may not be an appropriate smoothed analysis model for the discrete setting (for example, the perturbed vector $\mathcal{O}_i$ could have negative entries). 
\anote{11/28: added some detail.}

Using a trick from \cite{AMR09, AHK12}, we will translate the problem into the setting of multi-view models. Let $m = 2\ell + 1$ for some $\ell$ to be chosen later, and use the hidden state $Z_{\ell + 1}$ as the latent variable. \anote{11/28 added:}
In what follows, we will abuse notation and also represent the states using the standard basis vectors $e_1, e_2, \dots, e_r \in \R^r$: for each $j \in [r], \ell \in [m]$,  $Z_\ell=e_j \in \R^r$ iff the state at time $\ell$ is $j$. Our three views are obtained by looking at the past, present, and future observations: the first view is $X_\ell  \otimes X_{\ell-1} \otimes \ldots \otimes X_1$, the second is $X_{\ell + 1}$ and the third is $X_{\ell + 2} \otimes X_{\ell+3} \otimes \ldots X_{2\ell + 1}$. We can access these views by viewing the moment tensor $X_1 \otimes \ldots \otimes X_{2\ell+1}$ as a 3-tensor of shape $n^\ell \times n \times n^\ell$. The conditional expectations of these three views are given by matrices $A$, $B$, and $C$ of dimensions $n^\ell \times r$, $n \times r$, and $n^\ell \times r$ respectively. Explicitly, these matrices satisfy
\begin{align*}
\E[X_\ell \otimes \ldots \otimes X_1 | Z_{\ell + 1}] &= AZ_{\ell+1},\\
\E[X_{\ell+1} | Z_{\ell+1}] &= BZ_{\ell+1},\\ \E[X_{\ell+2}\otimes \ldots \otimes X_{2\ell+1} | Z_{\ell+1}] &= CZ_{\ell+1}.
\end{align*}
Let $P' = \diag(w)P^T\diag(w)^{-1}$, which is the reverse transition matrix $Z_i \to Z_{i-1}$, and let $X\odot Y$ denote the Khatri-Rao product of $X$ and $Y$, given in terms of its columns by $(X\odot Y)_i = X_i \otimes Y_i$. 
Then we can write down $A$, $B$, and $C$ in terms of the transition and observation matrices as follows. This fact is straightforward to check, so we leave the details to \cite{AMR09}.
\begin{align}
A &= (( \ldots (\tOcal P') \odot \tOcal)P')\odot \tOcal) \ldots P')\odot \tOcal)P'\\
B &= \tOcal\\
C &= (( \ldots (\tOcal P) \odot \tOcal)P)\odot \tOcal) \ldots P)\odot \tOcal)P,
\end{align}
where $\tOcal$ and $P$ or $P'$ appear $\ell$ times each in $A$ and $C$. Our goal is to upper bound the condition numbers of $A$ and $C$. Once we do this, we will be able to use a argument similar to that in \cite{BCV} to obtain $P$ and $\tOcal$ up to an inverse polynomial error.

The proof of this theorem will use a simple lemma relating the minimum singular value of a matrix $A$ to that of a matrix obtained by adding together rows of $A$.

\begin{lemma} \label{lem:addingmatrixrows}
Let $n_1, n_2, n_3$ be positive integers with $n_2 \ge n_3$. Let $A = (A_{(i_1,i_2),j}) \in \R^{n_1n_2 \times n_3}$ be a matrix, and let $B \in \R^{n_2 \times n_3}$ be the matrix whose $i_2$th row is $\sum_{i_1} A_{[(i_1, i_2),:]}$. Then $\sigma_{n_3}(A) \ge \frac{1}{\sqrt{n_1}} \sigma_{n_3}(B)$.
\end{lemma}
\begin{proof}
We can write $B = MA$, where $M \in \R^{n_2 \times n_1n_2}$ is a matrix whose $i$th row consists of $n_1(i-1)$ zeros, then $n_1$ ones, then $n_1(n_2-i)$ zeros. For any $v = (v_{ij}) \in \R^{n_1n_2}$, applying the Cauchy-Schwarz inequality gives
\[
\|Mv\|^2 = \sum_{i=1}^{n_2} (M_{[i,:]} \cdot v)^2 = \sum_{i=1}^{n_2}\left(\sum_{j=1}^{n_1} v_{ij}\right)^2 \le n_1 \|v\|^2.
\]
Therefore $\sigma_{max}(M) \le \sqrt{n_1}$. Since $\sigma_{min}(B) \le \sigma_{max}(M)\sigma_{min}(A)$, we have $$\sigma_{min}(A) \ge \frac{1}{\sqrt{n_1}}\sigma_{min}(B).$$
\end{proof}

\begin{theorem}
Let $\ell \in \Z_+$ be a constant. Suppose we are given a Hidden Markov Model in the setting described above, satisfying the following conditions:
\begin{enumerate}
\item $P\in \R^{r\times r}$ is $d$-sparse, where $d < O(\min\{n/\ell^2, n/r^{1/\ell}\})$ and $n = \Omega(r^{1/\ell})$. In addition, we assume $\sigma_{min}(P) \ge \gamma_1$. 
\item The columns of $\Ocal \in \R^{n\times r}$ are polynomially bounded (i.e.~the lengths are bounded by some polynomial in $n$) and are perturbed by independent Gaussian noise $N(0,\rho^2/n)^n$ to obtain $\tOcal$, with columns $\{\tOcal_i\}$.
\item The stationary distribution $w$ of $P$ has $w_i > \gamma_2$ for all $i \in [r]$.
\end{enumerate}
Then there is an algorithm that recovers $P$ and $\tOcal$ up to $\epsilon$ error (in the Frobenius norm) with probability at least $1-\exp(-\Omega_\ell(n))$,  using samples of $m = 2\ell + 1$ consecutive observations of the Markov chain. The algorithm runs in time $(n/(\rho\gamma_1\gamma_2\epsilon))^{O(\ell)}$.
\end{theorem}
\begin{proof}
We will show that $C$ is well-conditioned. First note that since the columns of $\tOcal$ (and therefore of $C$) are polynomially bounded, $\sigma_{max}(C)$ is also bounded by some polynomial in $n$ and $r$. Therefore we only need to give a lower bound on $\sigma_{min}(C)$. Since $\sigma_{min}(P') \ge \gamma_2\cdot\sigma_{min}(P)$, the proof for $A$ is identical. We can write $C = M(\tOcal, P)\cdot F(P)$, where $M \in \R^{n^\ell \times R}$ is a matrix whose columns are order-$\ell$ tensor products of $\{\tOcal_i\}$ and $F(P) \in \R^{R\times r}$ is a matrix of coefficients. We will show that each of these factors is well-conditioned, which will give us a bound on the condition number of $C$.

First we work with $M$. The columns of $M$ are all of the tensor products of $\{\tOcal_i\}$ that appear in $C$. Specifically, if the columns of $\tOcal$ are $\{\tOcal_i\}_{i \in [r]}$, then the columns of $M$ are all tensor products of the form 
\begin{equation}
\tOcal_{i_1} \otimes \ldots \otimes \tOcal_{i_\ell},
\end{equation}
where $P_{i_s,i_{s+1}} \neq 0$ for all $s = 1, \ldots, \ell-1$. The key here is that while the noise coming from the $\rho$-perturbations of $\{\Ocal_i\}$ is not independent column to column, any column of $M$ has noise that is highly correlated with only a few other columns. 

In order to apply Theorem~\ref{thm:indepmonomials}, we need to find $\Delta_1, \ldots, \Delta_\ell$. Fix a column $M_i$ of $M$. For $s < \ell$, we have
\begin{equation}
\Delta_s(M_i) \le \binom{\ell}{s}d^s.
\end{equation}
To show why, we describe a way of generating all columns of $M$ that differ from $M_i$ in $s$ factors. First, choose a set $S \subseteq [\ell]$ with $|S| = s$, which will specify the places 
at which the new column will differ from $M_i$. Begin at one place at which the new column will not differ, which is possible because $s < \ell$. Fill in the remaining factors by progressing by step forwards and backwards until each factor is chosen. Each time a place in $S$ is encountered, we have at most $d$ choices due to the sparsity of $P$.
\begin{remark}
Note that not all of these choices may correspond to a path through the state space of the Markov chain. Thus additional conditions limiting the number of short cycles in the graph of the Markov chain could lead to smaller upper bounds on $\{\Delta_s\}$.
\end{remark}
For $s = \ell$, we have $\Delta_\ell(M_i) \le R \le r \cdot d^\ell$ since all of the $\ell$ factors are arbitrary as long as they determine a path in the Markov chain.

Now the condition of Theorem~\ref{thm:indepmonomials} becomes
\begin{equation}
rd^\ell + \sum_{s=1}^{\ell-1} \binom{\ell}{s}d^s\left(\frac{n}{\ell}\right)^{\ell-s} \le c\left(\frac{n}{\ell}\right)^\ell \qquad \text{for } c \in (0,1),
\end{equation}
which holds by the restrictions on $d$ and $r$. Therefore we conclude that $\sigma_{min}(M) \ge \Omega_\ell(1)\cdot(\rho/n)^\ell/\sqrt{R}$ with probability at least $1 - \exp(-\Omega_\ell(n) + \log R) \ge 1 - \exp(-\Omega_\ell(n) + \log n^\ell) = 1 - \exp(-\Omega_\ell(n)).$

Next, we show that $F$ is well-conditioned. To simplify notation, we write as if $R = r^\ell$ (in which case $M$ would have many unused columns and $F$ would have many zero rows and columns). Index the rows of $F$ by a tuple $(i_1, \ldots, i_\ell)$. We have
\begin{equation}
    F_{(i_1, \ldots, i_\ell), j} = P_{ji_1}P_{i_1i_2}\cdots P_{i_{\ell-1}i_\ell}.
\end{equation}
In other words, the coefficient of $\tOcal_{i_1} \otimes \ldots \otimes \tOcal_{i_\ell}$ in column $j$ of $C$ is the probability, given that you begin at state $j$, of traveling through states $i_1, \ldots, i_\ell$. 

We want to give a lower bound for the least singular value of $F$. Lemma~\ref{lem:addingmatrixrows} shows that it is enough to bound the least singular value of a matrix obtained by adding together rows of $F$. Using this idea, we sum over all rows with the same $i_\ell$ to obtain a matrix $F' \in \R^{r \times r}$ with entries
\begin{equation}
    F'_{i,j} = \sum_{i_1, \ldots, i_{\ell-1}} P_{ji_1}P_{i_1i_2}\cdots P_{i_{\ell-1}i}.
\end{equation}
Thus we have $F' = (P^\ell)^T$, which has $\sigma_{min}(F') \ge \gamma_1^\ell$. Therefore Lemma~\ref{lem:addingmatrixrows} gives $\sigma_{min}(F) \ge \gamma_1^\ell/r^{\ell/2}$.

These two results show that 
$$\sigma_{min}(C) \ge \Omega_\ell(1)\cdot(\rho\gamma_1)^\ell/(n\sqrt{rd})^\ell r^{1/2} \ge \Omega_\ell(1)\cdot\left(\frac{\rho\gamma_1}{\sqrt{n^3r}}\right)^\ell$$
with probability at least $1 - \exp(-\Omega_\ell(n)).$

As mentioned above, we also get $\sigma_{min}(A) \ge \Omega_\ell(1)\cdot(\rho\gamma_1\gamma_2/\sqrt{n^3r})^\ell$ with the same probability. In order to recover $P$ and $\tOcal$, we use an algorithm similar to Algorithm 1 from Sharan et al.~\cite{Sharanetal}. First, we can estimate the moment tensor $X_1 \otimes \ldots \otimes X_{2\ell+1}$ to sufficient accuracy using $\poly_\ell(n, 1/\epsilon)$ many samples since each observation vector has a conditional distribution which is subgaussian. This follows from standard large deviation bounds, for example see Lemma C.1 in \cite{BCV}. Next, we can  obtain $A$, $B$, and $C$ up to an error $\delta = \poly(\epsilon, n, \rho)$ using a tensor decomposition algorithm such as in \cite{BCMV}.
Since $B = \tOcal$, it only remains to find $P$. To do this, we use a similar trick to \cite{AMR09}. We will use the fact that $C$ and $P$ are both well-conditioned. First, let $D = (C \odot \tOcal)P$. Note that we can obtain $D$ by following the entire procedure again but increasing $\ell$ by one. Since we already have $\tOcal$ up to a small error, we can also find $C \odot \tOcal$. Now $\sigma_{min}(C \odot \tOcal) \ge \sigma_{min}(D)/\sigma_{max}(P)$, and $\sigma_{max}(P) \le \sqrt{r}$. Therefore we can recover $P$ from $D$ and $C \odot \tOcal$ up to the required inverse polynomial error. 

\end{proof}

%% file: FOOBI.tex

In this section, we describe an algorithm to decompose $2\ell$'th order tensors of rank up to $n^{\ell}$. Let us start by recalling the problem: suppose $A_1, \dots, A_R$ are vectors in $\R^n$. Consider the $2\ell$'th order moment tensor
\[ M_{2\ell} = \sum_{i=1}^R A_i^{\ot 2\ell}. \]

The tensor decomposition problem asks to find the vectors $A_i$ to a desired precision (up to a re-ordering), given only the tensor $M_{2\ell}$. The question of {\em robust recovery} asks to find the vectors $A_i$ to a desired precision given access to a {\em noisy} version of $M_{2\ell}$, specifically, given only the tensor $T = M_{2\ell} +\err$.  The aim is to show that recovery is possible, assuming that $\norm{\err}$ is bounded by some polynomial in $n$ and the desired precision for recovering the $A_i$. We give an algorithm for robust recovery, under certain condition number assumptions on the $A_i$. Then using the methods developed earlier in the paper, we show that these assumptions hold in a smoothed analysis model. 

\subsection{Robust decomposition assuming non-degeneracy}\label{sec:foobi-robust}
We will now consider a generalization of the algorithm of Cardoso~\cite{Cardoso}, and prove robust recovery guarantees under certain non-degeneracy assumptions.  As stated in the introduction, our contribution is along two directions: the first is to extend the algorithms of~\cite{Cardoso} and~\cite{dlVKKV05} to the case of $2\ell$'th order tensors. Second (and more importantly), we give a robustness analysis. 

We now define an operator, and then a matrix whose condition number is important for our argument. Given $\ell$'th order tensors $X, Y$, we define the operator $\Phi$ as $\Phi(X, Y) = \Psi(X, Y) + \Psi(Y, X)$, where $\Psi: \R^{\ell} \times \R^{\ell} \mapsto \R^{2\ell}$ is defined by:
\begin{equation}\label{eq:def-psi}
\Psi(X, Y) (i_1, i_2, \dots, i_\ell, j_1, j_2, \dots, j_\ell) = X_{i_1 \dots i_{\ell-1} i_\ell} Y_{j_1 \dots j_{\ell-1} j_\ell} - X_{i_1 \dots i_{\ell-1} j_\ell} Y_{j_1 \dots j_{\ell-1} i_\ell}
\end{equation}
One of the nice properties of $\Phi$ above is that $\Phi(X, X) = 0$ for a {\em symmetric} tensor\footnote{An $\ell$'th order tensor $T$ is said to be symmetric if $T_{i_1 i_2 \dots i_\ell} = T_{\pi(i_1) \pi(i_2) \dots \pi(i_\ell)}$ for any permutation $\pi$.} $X$ iff $X = \bu^{\otimes \ell}$, for some $\bu \in \R^n$ (and for this reason,~\cite{Cardoso} who introduced such an operator for $\ell=2$ and subsequent works refer to this as a rank-1 ``detector''). The algorithm and its analysis only use the easy direction of the above statement, namely $\Phi (\bu^{\otimes \ell}, \bu^{\otimes \ell}) = 0$ for any $\bu \in \R^n$, and thus we do not prove the property above.

The following matrix plays a crucial role in the analysis: consider the $\binom{R}{2}$ vectors of the form $\Phi(A_i^{\otimes \ell}, A_j^{\otimes \ell})$, for $i < j$. Let $M_\Phi$ be the matrix with all of these vectors as columns.  Thus $M_{\Phi}$ is of dimensions $n^{2\ell} \times \binom{R}{2}$.

\nvgap
\paragraph{Relevant condition numbers.}  Our robustness analysis will depend on (a) the condition number of the matrix $U:= A^{\odot \ell}$, which we will denote by $\kappa_U$, and (b) the condition number of the matrix $M_\Phi$ described above, which we will denote by $\kappa_M$. For convenience, let us also define $\bu_i = A_i^{\ot \ell}$, flattened. From the definition of $U$ above, we also have $M_{2\ell}$ equal to $UU^T$, when matricized. 

The following is our main result of the section. 

\begin{theorem}\label{thm:foobi-robust}
Given the tensor $T = M_{2\ell} + \err$, an accuracy parameter $\eps$, and the guarantee that $\norm{\err}_{F} \le \eps^c / (\kappa_U \kappa_M)^{c'}$ for some constants $c, c'$, there is an algorithm that outputs, with failure probability $1-\gamma$, a set of vectors $\{B_i\}_{i=1}^R$ such that 
\[ \min_{\pi} \sum_i \norm{A_i - B_{\pi(i)}} \le \eps.\]
Furthermore, this algorithm runs in time $\poly(n^\ell, \kappa_U, \kappa_M, \log(1/\gamma))$.
\end{theorem}
\pnote{I'm fairly certain there's no $\eps$ dependence in the running time, but someone else should make sure as well.}

\nvgap
\paragraph{Remark.}  We note that the above statement does not explicitly require a bound on the rank $R$.  However, the finiteness of the condition numbers $\kappa_U$ and $\kappa_M$ implies that $R \le n^{\ell}/2$. Our theorem~\ref{thm:foobi-smooth} shows that when $R \le n^{\ell}/2$, the condition numbers are both polynomial in $n$ in a smoothed analysis model. Also, we do not explicitly compute $c, c'$. From following the proof na\"ively, we get them to be around 8, but they can likely be improved.

\subsubsection{Outline of the proof and techniques}
We will start (section~\ref{sec:foobi-no-noise}) by presenting the FOOBI procedure for arbitrary $\ell$. The algorithm proceeds by considering the top eigenvectors of the matricized version of $M_{2\ell}$, and tries to find {\em product vectors} (i.e. vectors of the form $x^{\ot \ell}$) in their span. This is done via writing a linear system involving the basis vectors. 

In section~\ref{sec:foobi-robust-sub}, we show that the entire procedure can be carried out even if $M_{2\ell}$ is only known up to a small error. The technical difficulty in the proof arises for the following reason: while a small perturbation in $M_{2\ell}$ does not affect the top-$R$ SVD of the (matricized) $M_{2\ell}$, if we have no guarantees on the gaps between the top $R$ eigenvalues, the eigen\emph{vectors} of the perturbed matrix can be quite different from those of $M_{2\ell}$. Now the FOOBI procedure performs non-trivial operations on these eigenvectors when setting up the linear system we mentioned in the previous paragraph. Showing that the solutions are close despite the systems being different is thus a technical issue we need to overcome.

\subsubsection{Warm-up: the case of $\err = 0$}\label{sec:foobi-no-noise}

Let us start by describing the algorithm in the zero error case. This case illustrates the main ideas behind the algorithm and generalizes the FOOBI procedure to arbitrary $\ell$.

The algorithm starts by computing the SVD of the matricized $M_{2\ell}$ (i.e., $UU^T$). Thus we obtain matrices $E$ and $\Lambda$ such that $UU^T = E \Lambda E^T$. Let $H$ denote the matrix $E \Lambda^{1/2}$.  Then we have $HH^T = UU^T$, and thus there exists an orthogonal $R \times R$ matrix $Q$ such that $U = HQ$.  Thus, finding $U$ now reduces to finding the orthogonal matrix $Q$. 

This is done using in a clever manner using the rank-1 detecting device $\Phi$. Intuitively, if we wish to find one of the columns of $U$, we may hope to find a linear combination $\sum_j \alpha_j H_j$ of the $\{H_j\}$ such that $\Phi (\sum_j \alpha_j H_j, \sum_j \alpha_j H_j) = 0$.  Each column of $Q$ would provide a candidate solution $\alpha$.  However, this is a quadratic system of equations in $\alpha_i$, and it is not clear how to solve the system directly. 

The main idea in~\cite{Cardoso} is to find an alternate way of computing $Q$. The first observation is that $\Phi$ is bi-linear (i.e., linear in its arguments $X, Y$). Thus, we have $ \Phi(\sum_j \alpha_j H_j, \sum_j \alpha_j H_j) = \sum_{i,j \in [R]} \alpha_i \alpha_j \Phi (H_i, H_j)$.  Now, consider the {\em linear} system of equations
\begin{equation}\label{eq:foobi-system}
\sum_{i,j \in [R]} W_{ij} \Phi(H_i, H_j) = 0.
\end{equation}

This is a system of $n^{2\ell}$ equations in $R^2$ variables.  The reasoning above shows that for every column $Q_i$ of $Q$, we have that $W = Q_i Q_i^T$ is a solution to \eqref{eq:foobi-system}.  Because of linearity, this means that for any diagonal matrix $\bD$, $Q \bD Q^T$ is a solution to the linear system as well. The main observation of~\cite{Cardoso} is now that any symmetric solution $W$ (i.e. one that satisfies $W_{ij} = W_{ji}$) is of this form! Thus, the matrix $Q$ can be computed by simply finding a ``typical'' symmetric solution $W$ and computing its eigen-decomposition. Let us now formalize the above.

\begin{lemma}\cite{Cardoso}\label{lem:lathauwer}
The space of symmetric solutions to the system of equations~\eqref{eq:foobi-system} has dimension precisely $R$, and any solution is of the form $W = Q \bD Q^T$, for some diagonal matrix $\bD$.
\end{lemma}
\begin{proof}
Consider any symmetric solution $W$.  Because of bi-linearity, using the fact that $H Q = U$, or $H = UQ^T$, we have that $H_i = \sum_{s} U_s (Q^T)_{si} = \sum_s U_s Q_{is}$.  Thus for any $i,j$,
\[ \Phi(H_i, H_j) = \sum_{s,t} Q_{is} Q_{jt} \cdot \Phi(U_s, U_t). \]

Thus, 
\begin{equation}\label{eq:foobi-system-zero}
\sum_{i,j} W_{ij} \Phi(H_i, H_j) = \sum_{s, t} \Phi(U_s, U_t) \cdot \left( \sum_{i,j} W_{ij} Q_{is} Q_{jt} \right) = \sum_{s,t} \Phi(U_s, U_t) \iprod{W, Q_s Q_t^T} .
\end{equation}

Since $\kappa_M < \infty$, we have that $\{\Phi(U_s, U_t) : s < t\}$ is linearly independent. Now, since $\Phi(U_s, U_t) = \Phi(U_t, U_s)$, and since $\Phi(U_s, U_t) \ne 0$ for all $s \ne t$ (the latter is a simple computation, using the fact that $A_s \ne A_t$), we must have that
\[ \text{ for all } s \ne t,~~ \iprod{W, Q_s Q_t^T} = 0. \]

Now, since $Q$ is an orthogonal matrix, we have that $\{ Q_s Q_t^T \}_{s, t \in [R]}$ forms an orthonormal basis for all $R \times R$ matrices.  The above equality thus means that $W$ lies only in the span of $\{ Q_s Q_s^T \}_{s \in [R]}$. This implies that $W = Q \bD Q^T$, for some diagonal matrix $\bD$.

Plugging back into~\eqref{eq:foobi-system-zero}, we see that any $W$ of this form satisfies the equation. As the $Q_s Q_s^T$ are all orthogonal, we have found a solution space of dimension precisely $R$.
\end{proof}

To handle the robust case, we also need a slight extension of the lemma above.  Let $H_{\Phi}$ denote a matrix that has $R(R+1)/2$ columns, described as follows. The columns correspond to pairs $i,j \in [R]$, for $i\le j$.  For $i = j$, the corresponding column is $\Phi(H_i, H_i)$ and for $i< j$, the corresponding column is $\sqrt{2} \cdot \Phi(H_i, H_j)$. We note that the null space of $H_{\Phi}$ can be mapped in a one-one manner to symmetric $R\times R$ matrices $W$.  For any $z = (z_{ij})_{i \le j}$, define the symmetric $R \times R$ matrix $\psi(z)$ to have $\psi(z)_{ii} = z_{ii}$ and $\psi(z)_{ij} = \psi(z)_{ji} = \frac{z_{ij}}{\sqrt{2}}$.  The point of this definition is that $\iprod{z, z'} = \iprod{ \psi(z), \psi(z') }$.  Note that $\psi^{-1}$ is well-defined, and that it takes symmetric matrices to $R(R+1)/2$-dimensional vectors (and preserves dot-products).

Further, we have 
\begin{align}
H_{\Phi} z &= \sum_{i} \Phi(H_i, H_i) z_{ii} + \sum_{i < j} \sqrt{2} \cdot \Phi(H_i, H_j) z_{ij} \notag \\
&= \sum_{i} \Phi(H_i, H_i) z_{ii} + \sum_{i<j} 2 \Phi(H_i, H_j) \psi(z)_{ij} \notag\\  
&= \sum_{i,j} \Phi (H_i, H_j) \psi(z)_{ij}.\label{eq:psi-main}
\end{align}

Using this correspondence, Lemma~\ref{lem:lathauwer} implies that $H_{\Phi}$ has a null space of dimension precisely $R$ (corresponding to the span of $\psi^{-1} (Q_s Q_s^T)$, for $s \in [R]$). We now claim something slightly stronger.

\begin{lemma}\label{lem:key-sigma-min}
Let $\lambda$ denote the $(R+1)$th smallest singular value of $H_\Phi$. We have that $\lambda \ge \sigma_{\min} (M_{\Phi})$.  Recall that $M_{\Phi}$ was defined to be the matrix with columns $\Phi (A_i^{\otimes \ell}, A_j^{\otimes \ell})$, for $i < j$.
\end{lemma}

\begin{proof}
Consider any $z$ orthogonal to span$\{ \psi^{-1} (Q_s Q_s^T) : s \in [R] \}$.  Then, $\psi(z)$ is orthogonal to $Q_s Q_s^T$ for all $s$, as $\psi$ preserves dot-products. Thus, using our earlier observation that $\{Q_s Q_t^T\}$ forms an orthonormal basis for all $R\times R$ matrices, we can write
\[ \psi(z) = \sum_{s \ne t} \alpha_{st} Q_s Q_t^T. \]

Since $\psi(z)$ is symmetric, we also have $\alpha_{st} = \alpha_{ts}$.  Now, using the expansion~\eqref{eq:foobi-system-zero} with $W_{ij} = \psi(z)_{ij}$, we have 
\[ \sum_{i,j} \psi(z)_{ij} \Phi(H_i, H_j) = \sum_{s, t} \alpha_{st} \Phi (U_s, U_t) = 2\sum_{s < t} \alpha_{st} \Phi (U_s, U_t). \]

Combining this with \eqref{eq:psi-main}, and the definition of the smallest singular value, we obtain
\[ \norm{ H_{\Phi} z }_F^2 \ge 4 \left( \sum_{s < t} \alpha_{st}^2 \right) \sigma_{\min} (M_{\Phi})^2. \]

Finally, since $\norm{z}^2 = \norm{\psi(z)}_F^2 = 2 \sum_{s < t} \alpha_{st}^2$, the desired conclusion follows (indeed with a slack of a factor $\sqrt{2}$).
\end{proof}

The following theorem then gives the algorithm to recover $Q$, in the case $\err = 0$.
\begin{theorem}\label{thm:foobi-algo-noerror}
Let $S$ be the subspace (of $\R^{R \times R}$) of all symmetric solutions to the system of equations $\sum_{ij} W_{ij} \Phi (H_i, H_j) = 0$. Let $Z$ be a uniformly random Gaussian vector in this subspace of unit variance in each direction. Then with probability at least $9/10$, we have that
\[ Z = \sum_{i} \alpha_i Q_i Q_i^T, \text{ where } \min_{i \ne j} |\alpha_i - \alpha_j| \ge \frac{1}{20 R^2}.\]
Thus the SVD of $Z$ efficiently recovers the $Q_i$, with probability $\ge 9/10$.
\end{theorem}
\begin{proof}
From the lemmas above, we have that the space $S$ is precisely the span of $Q_s Q_s^T$, for $s \in [R]$.  These are all orthogonal vectors.  Thus a random Gaussian vector in this space with unit variance in each direction is of the form $\sum_i \alpha_i Q_i Q_i^T$, where the $\alpha_i$ are independent and distributed as the univariate Gaussian $\mathcal{N}(0, 1)$.  

Now, for any $i, j$, we have that $\alpha_i - \alpha_j$ is distributed as $\mathcal{N}(0, 2)$, and thus
\[ \Pr \left[ |\alpha_i - \alpha_j| \le \frac{1}{20 R^2} \right] \le \frac{1}{20 R^2}. \]

Taking a union bound over all pairs $i, j$ now gives the result. 
\end{proof}

This completes the algorithm for the case $\err = 0$.  Let us now see how to extend this analysis to the case in which $\err \ne 0$.

\subsubsection{A robust analysis}\label{sec:foobi-robust-sub}
We will now prove an approximate recovery bound by following the above analysis, when $\err$ is non-zero (but still small enough, as in the statement of Theorem~\ref{thm:foobi-robust}). As is common in such analyses, we will use the classic Davis-Kahan Sin-$\theta$ theorem. We start by recalling the theorem. To do so, we need some notation.  

Suppose $V_1$ and $V_2$ are two $n \times d$ matrices with orthonormal columns.  Then the matrix of {\em principal angles} between the column spans of $V_1$ and $V_2$ is denoted by $\Theta(V_1, V_2)$, and is defined to be the diagonal matrix whose entries are $\arccos (\lambda_i)$, where $\lambda_i$ are the singular values of $V_1^T V_2$. 

\begin{theorem}[Sin-$\theta$ theorem,~\cite{Davis1970}]\label{thm:davis-kahan}
Let $\Sigma$ and $\Sigma' \in \R^{n \times n}$ be symmetric, with eigenvalues $\lambda_1 \ge \lambda_2 \ge \dots \lambda_n$ and $\lambda_1' \ge \lambda_2' \ge \dots \ge \lambda_n'$.  Let $1 \le r \le s \le n$, and let $d = s-r+1$.  Let $V$ be a matrix with columns being the eigenvectors corresponding to $\lambda_r \dots \lambda_s$ in $\Sigma$, and suppose $V'$ is similarly defined. Let 
\[ \delta := \inf \{ |\lambda' -\lambda| : \lambda \in [\lambda_s, \lambda_r], \lambda' \in (-\infty, \lambda_{s+1}'] \cup [\lambda_{r-1}', \infty) \}, \]
which we assume is $>0$.  Then we have
\[ \norm{\sin \Theta(V, V')}_F \le \frac{\norm{\Sigma - \Sigma'}_F}{\delta}.\]
Furthermore, there exists an orthogonal matrix $O'$ such that
\begin{equation}\label{eq:davis-kahan-basis}
\norm{V - V'O'}_F \le \frac{\sqrt{2} \norm{\Sigma - \Sigma'}_F}{\delta}.
\end{equation}
\end{theorem}

We note that the precise statement above is from~\cite{Yu2015}. Our proof will follow the outline of the $\err = 0$ case.  The first step is to symmetrize the matricized version of $T$, so that we can take the SVD.  We have the following simple observation.

\begin{lemma}
Let $A \in \R^{n \times n}$, and define $A' = (A+A^T)/2$.  Let $B \in \R^n$ be symmetric. Then $\norm{A' - B}_F \le \norm{A - B}_F$.
\end{lemma}
\begin{proof}
The lemma follows from observing that $A'$ is the projection of $A$ onto the linear space of symmetric $n \times n$ matrices, together with the fact that projections to convex sets only reduces the distance.
\end{proof}

Let $T'$ be the symmetric version of the matricized version of $T$.  Then we have $\norm{T' - UU^T}_F \le \norm{\err}_F$.  Likewise, let $\wh{T}$ be the projection of $T'$ onto the PSD cone (we obtain $\wh{T}$ by computing the SVD and zero'ing out all the negative eigenvalues).  By the same reasoning, we  have $\norm{\wh{T} - UU^T}_F \le \norm{\err}_F$.  For convenience, in what follows, we denote $\norm{\err}_F$ by $\eta$. 

Next, we need a simple lemma that relates the error in a square root to the error in a matrix.

\begin{lemma}\label{lem:err-square-root}
Let $Z$ and $H$ be $n \times d$ matrices with $d \le n$, and suppose $\norm{ZZ^T - HH^T} \le \delta$.  Then there exists an orthogonal matrix $Q$ such that
\[ \norm{ZQ - H}_F \le (d \delta)^{1/2} + \frac{2\delta d \sigma_1 (H)}{\sigma_d (H)^2}. \]
\end{lemma}

\begin{proof}
Let $ZZ^T = M_1 \Sigma_1 N_1^T$, and let $HH^T = M_2 \Sigma_2 N_2^T$, where $M_i, N_i$ are $n \times d$ matrices with orthonormal columns.  Now, the theory of operator-monotone functions acting on PSD matrices (see e.g.~\cite{Bhatia}, Theorem X.1.1) implies that
\[ \norm{M_1 \Sigma_1^{1/2} N_1^T - M_2 \Sigma_2^{1/2} N_2^T}_F \le \delta^{1/2}. \]

Now we may apply the Sin-$\theta$ theorem (with $r=1$ and $s=d$ in the statement above) to conclude that there exists an orthogonal matrix $Q_1$ such that $\norm{N_1 Q_1 - N_2}_F \le \frac{\sqrt{2} ~ \delta}{\sigma_d (H)^2}$.  Thus, writing $N_2 = N_1 Q_1 + \Delta$, the LHS above becomes
\[ \norm{M_1 \Sigma_1^{1/2} N_1^T - M_2 \Sigma_2^{1/2} Q_1^T N_1^T - M_2 \Sigma_2^{1/2} \Delta^T}_F. \] 

Now, we have $\norm{M_2 \Sigma_2^{1/2} \Delta^T}_F \le \norm{ M_2 \Sigma_2^{1/2}}_F \norm{\Delta}_F$.  The first term is simply $(\tr (\Sigma_2))^{1/2} \le d^{1/2} \sigma_1 (H)$.  Using this, we obtain
\begin{align*}
\norm{ (M_1 & \Sigma_1^{1/2} - M_2 \Sigma_2^{1/2} Q_1^T) N_1^T }_F \\
	& = \norm{M_1 \Sigma_1^{1/2} N_1^T - M_2 \Sigma_2^{1/2} N_2^T + M_2 \Sigma_2^{1/2} \Delta^T}_F \\
	&\le \delta^{1/2} + \frac{ 2\delta d^{1/2} \sigma_1 (H)}{\sigma_d (H)^2}.
\end{align*}

We can now appeal to the simple fact that for a matrix $X$, for any $N_2$ with orthonormal columns, we have $\norm{X}_F = \norm{X N_2^T N_2} \le \norm{X N_2^T}_F d^{1/2}$.  This gives us
\[ \norm{M_1 \Sigma_1^{1/2} - M_2 \Sigma_2^{1/2} Q_1^T}_F \le (d \delta)^{1/2} + \frac{2\delta d \sigma_1 (H)}{\sigma_d (H)^2}. \]

Thus, since $Z = M_1 \Sigma_1^{1/2} Q'$ for an orthogonal matrix $Q'$ and likewise for $H$, and because the product of orthogonal matrices is orthogonal, we have the desired result.
\end{proof}

In what follows, to simplify the notation, we introduce the following definition.
\begin{definition}[Poly-bounded function]
We say that a function $f$ of a parameter $\eta$ is {\em poly-bounded} if $f(\eta)$ is of the form $\eta^c \cdot \poly(n, R, \kappa_U, \kappa_{M})$, where $c > 0$ is a constant.
\end{definition}

Intuitively speaking, by choosing $\eta$ to be ``polynomially small'' in $n, R$ and the condition numbers $\kappa_U, \kappa_{M}$, we can make $f(\eta)$ arbitrarily small.  

Now, the lemma above gives the following as a corollary.
\begin{corollary}
Let $\wh{E} \wh{\Lambda} \wh{E}^T$ be the rank-$R$ SVD of $\wh{T}$, and let $UU^T = E \Lambda E^T$ be the SVD as before. Define $\wh{H} = \wh{E} \wh{\Lambda}^{1/2}$ and $H = E\Lambda^{1/2}$.  Then there exists an orthogonal matrix $P$ such that $\norm{\wh{H}P - H}_F \le f_1 (\eta)$ for some poly-bounded function $f_1$.
\end{corollary}
\begin{proof}
The desired conclusion follows from Lemma~\ref{lem:err-square-root} if we show that $\norm{ \wh{E} \wh{\Lambda} \wh{E}^T - E \Lambda E^T} \le 2\eta$.  This follows from the fact that $\norm{\wh{T} - \wh{E} \wh\Lambda \wh{E}^T} \le \eta$ (which is true because the SVD gives the closest rank-$k$ matrix to $\wh{T}$ --- and $UU^T$ is at distance at most $\eta$), together with the triangle inequality.
\end{proof}

Informally speaking, we have shown that $\wh{H} P \approx_\eta H$, for an orthogonal matrix $P$. We wish to now use our machinery from Section~\ref{sec:foobi-no-noise} to find the matrix $U$, which will then allow us to obtain the vectors in the decomposition. 

Let us define $H' = HP^T$, where $P$ is as above.  Thus we have $H = H' P$ (and thus $H' \approx_\eta \wh{H}$, informally).  Further, if $Q$ is the orthogonal matrix such that $U = HQ$ (as in Section~\ref{sec:foobi-no-noise}), we have $U = H' PQ$.

\paragraph{Outline of the remainder.}  We first sketch the rest of the argument.  The key idea is the following: suppose we run the whole analysis in Section~\ref{sec:foobi-no-noise} using the matrices $H'$ and $PQ$ instead of $H$ and $Q$, we obtain that the set of symmetric solutions to the system of equations $\sum_{i, j \in [R]} W_{ij} \Phi (H'_i, H'_j)$ is precisely the span of the matrices $(PQ)_s (PQ)_s^T$.  Thus, a random matrix in the space of symmetric solutions can be diagonalized to obtain $(PQ)_s$.  Using $U = H' (PQ)$, one can reconstruct $U$.  Now, we have access to $\wh{H}$ and not $H'$.  However, we can relate the space of symmetric {\em approximate solutions} to the perturbed system to the original one in a clean way.  Taking a random matrix in this space, and utilizing the ``gap'' in Theorem~\ref{thm:foobi-algo-noerror}, we obtain the matrix $PQ$ approximately.  This is then used to find $\wh{U}$ that approximates $U$, completing the argument.
  
  
\begin{lemma}
For any $i, j \in [R]$, we have
\[ \norm{ \Phi (H_i', H_j') - \Phi (\wh{H}_i, \wh{H}_j)} \le O \left( \norm{H_i' - \wh{H}_i} \norm{H_j'} + \norm{H_j' - \wh{H}_j} \norm{\wh{H}_i} \right).\]
\end{lemma}
\begin{proof}
\begin{align*}
 \norm{ \Phi (H_i', H_j') - \Phi (\wh{H}_i, \wh{H}_j)}  \le  \norm{ \Phi (H_i', H_j') - \Phi (\wh{H}_i, H_j')} + \norm{\Phi(\wh{H}_i, H_j' ) - \Phi(\wh{H}_i, \wh{H}_j)}. 
\end{align*}
The first term can be bounded by $2 \norm{H_j'} \norm{H_i' - \wh{H}_i}$, and so also the second term is bounded by $2 \norm{\wh{H}_i} \norm{H_i' - \wh{H}_i}$, which implies the lemma.
\end{proof}

Next, as in Section~\ref{sec:foobi-no-noise}, define the $R(R+1)/2$ dimensional matrices $\wh{H}_{\Phi}$ and $H'_{\Phi}$.  Specifically, these matrices have columns corresponding to pairs $1 \le i \le j \le R$, and for $i=j$, the corresponding column of $H'_{\Phi}$ is $\Phi(H_i', H_i')$ and for $i \ne j$, the column is $\sqrt{2} \cdot \Phi(H_i', H_j')$.   A simple corollary of the lemma above is that
\begin{equation}\label{eq:hat-H-frob}
\norm{ \wh{H}_{\Phi} - H'_{\Phi}}_F \le O\left( \norm{\wh{H} - H'}_F \cdot (\norm{\wh{H}}_F + \norm{H'}_F ) \right) = f_2(\eta),
\end{equation}
for some poly-bounded function $f_2$.  This follows from the lemma and corollary above, together with an application of the Cauchy-Schwarz inequality.  Next, we show the following.

\begin{lemma}
For $1 \le r \le R$, we have $\sigma_r (\wh{H}_{\Phi}) \le f_2(\eta) $.  Also, we have $\sigma_{R+1} (\wh{H}_{\Phi}) \ge \sigma_{\min} (M_{\Phi}) - f_2 (\eta)$.
\end{lemma}
\begin{proof}
The main idea, as mentioned in the outline, is to apply Lemma~\ref{lem:key-sigma-min} to $H'$. If $\lambda'$ denotes the $(R+1)$th smallest singular value of $H_\Phi'$, then this lemma implies that $H'_{\Phi}$ has $R$ zero singular values and $\lambda' \ge \sigma_{\min} (M_{\Phi})$.  Weyl's inequality\footnote{Recall that the inequality bounds the change in eigenvalues due to a perturbation of a matrix by the spectral norm (and hence also the Frobenius norm) of the perturbation.} now immediately implies the lemma.
\end{proof}

From now on, suppose that $\eta$ is chosen small enough that $f_2(\eta) < \frac{\sigma_{\min} (M_{\Phi})}{2}$.  Next, let us define the spaces $S'$ and $\wh{S}$ as in Theorem~\ref{thm:foobi-algo-noerror}: let $S'$ be the space of all symmetric solutions to the linear system
\[ \sum_{i,j} W_{ij} \Phi (H_i', H_j') =0.\]
Likewise, let $\wh{S}$ be the space of symmetric matrices $\psi(z)$ (see Section~\ref{sec:foobi-no-noise} for the definition of $\psi$), where $z$ is in the span of the $R$ smallest singular values of $\wh{H}_{\Phi}$.  The analog of Theorem~\ref{thm:foobi-algo-noerror} is the following.

\begin{theorem}\label{thm:foobi-algo-error}
Let $Z$ be a uniformly random Gaussian vector in $\wh{S}$, and suppose that $Z = G \Sigma G^T$ is the SVD of $Z$.  Then with probability $\ge 9/10$, we have $\norm{G - PQ}_F \le f_3 (\eta)$, for some poly-bounded function $f_3$.
\end{theorem}
\begin{proof}
The first step is to show that the spaces $S'$ and $\wh{S}$ are close.  This is done via the Sin-$\theta$ theorem, applied to the matrices $(H'_{\Phi})^T H'_{\Phi}$ and $\wh{H}_{\Phi}^T \wh{H}_{\Phi}$.  Let $T'$ and $\wh{T}$ be the spans of the smallest $R$ singular vectors of the two matrices.  By Theorem~\ref{thm:davis-kahan} and the bounds on $\sigma_{R+1}$, we have that there exist orthonormal bases $\Upsilon$ and $\wh{\Upsilon}$ for these spaces such that for some orthonormal matrix $Q'$, 
\[ \norm{ \Upsilon Q' - \wh{\Upsilon} }_F  \le \frac{ \norm{ (H'_{\Phi})^T H'_{\Phi} - \wh{H}_{\Phi}^T \wh{H}_{\Phi} }_F } {\sigma_{\min} (M_\Phi)^2 }.  \]

Now, appealing to the simple fact that for any two matrices $X, Y$, $\norm{X^T X - Y^T Y}_F \le \norm{ X^T (X-Y) + (X^T -Y^T) Y}_F \le \norm{X - Y}_F (\norm{X}_F + \norm{Y}_F)$, we can bound the quantity above by $f_4 (\eta)$ for some poly-bounded function $f_4$.

We can now obtain bases for $\wh{S}$ and $S'$ by simply applying $\psi$ to the columns of the matrices $\wh{\Upsilon}$ and $\Upsilon$ respectively. Let us abuse notation slightly and call these bases $\wh{S}$ and $S'$ as well. By properties of $\psi$, we have that
\begin{equation}\label{eq:S-rotate}
\norm{S' Q' - \wh{S}}_F \le  \norm{ \Upsilon Q' - \wh{\Upsilon} }_F \le f_4(\eta).
\end{equation}

Next, note that a random unit Gaussian vector in the space $\wh{S}$ can be viewed as first picking $v \in \mathcal{N}(0,1)^R$ and taking $\wh{S} v$.  Now, using Theorem~\ref{thm:foobi-algo-noerror} if we consider the matrix $S' v$ (which is a random Gaussian vector in the space $S'$), with probability at least $9/10$, we have an eigenvalue gap of at least $\frac{1}{20 R^2}$.  Thus, using this and~\eqref{eq:S-rotate}, together with the Sin-$\theta$ theorem (used this time with precisely one eigenvector, and thus the rotation matrix disappears), we have that $\norm{ G_i - (PQ)_i } \le 20 R^2 f_4 (\eta)$.  Summing over all $i$ (after taking the square), the theorem follows.
\end{proof}

We can now complete the proof of the main theorem of this section.

\begin{proof}[Proof of Theorem~\ref{thm:foobi-robust}]
Theorem~\ref{thm:foobi-algo-error} shows that the matrix $G$ gives a good approximation to the rotation matrix $(PQ)$ with probability $9/10$. (Since this probability is over the randomness in the algorithm, we can achieve a probability of $1-\gamma$ by running the algorithm $O(\log 1/\gamma)$ times.) We now show that $\wh{H} G \approx H' PQ$: 
\begin{equation}
\norm{\wh{H} G - H' PQ}_F \le \norm{\wh{H} (G - PQ) + (\wh{H} - H') PQ}_F \le f_5 (\eta).
\end{equation}

Note now that $H' PQ$ is precisely $U$! Thus the matrix $\wh{U} := \wh{H} G$ (which we can compute as discussed above) is an approximation up to an error $f_5 (\eta)$.  Finally, to obtain a column $U_i$ of $U$, we reshape $\wh{U}$ into an $n \times n^{\ell -1}$ matrix, apply an SVD, and output the top left-singular-vector. This yields an error $f_6(\eta)$, for some poly-bounded function of $\eta$.
\end{proof}
\pnote{added sentence about failure probability $1-\gamma$}


\subsection{Smoothed analysis}\label{sec:foobi-smooth}
Finally, we show that Theorem~\ref{thm:foobi-robust} can be used with our earlier results to show the following.

\begin{theorem}\label{thm:foobi-smooth}
Suppose $T = \sum_{i \in [R]} \tA_i^{\ot 2\ell} + E$, where $\{A_i\}$ have polynomially bounded length. Given an accuracy parameter $\eps$ and any $0 < \delta < 1/\ell^2$, with probability at least $1 - \exp(-\Omega_{\ell}(n))$ over the perturbation in $\tA$, there is an efficient algorithm that outputs a set of vectors $\{B_i\}_{i=1}^R$ such that 
\[ \min_{\pi} \sum_i \norm{\tA_i - B_{\pi(i)}} \le \eps, \] as long as $R \le \delta n^{\ell}$, and $\norm{E}_F \le \poly(1/n, \rho, \eps)$, for an appropriate polynomial in the arguments.
\end{theorem}

The proof of this theorem goes via the robust decomposition algorithm presented in Theorem~\ref{thm:foobi-robust}. In order to use the theorem, we need to bound the two condition numbers $\kappa_U$ and $\kappa_M$. Since the columns of $A$ are polynomially bounded, the columns of $U$ and $M_\Phi$ are as well, so $\sigma_{max}(U)$, $\sigma_{max}(M_\Phi)$ are bounded by some polynomial in $n$. Therefore we only need to give lower bounds on $\sigma_{min}(U)$ and $\sigma_{min}(M_\Phi)$. We now use Theorem~\ref{thm:indepmonomials} to prove that these quantities are both polynomially bounded with high probability in a smoothed analysis setting. This would complete the proof of Theorem~\ref{thm:foobi-smooth}. 

\begin{lemma}
Let $U = \widetilde{A}^{\odot \ell}$, and $M_{\Phi}$ be the matrix whose columns are indexed by pairs $i, j \le R$, and whose $\{i,j\}$'th column is $\Phi(\widetilde{A}_i^{\otimes \ell}, \widetilde{A}_j^{\otimes \ell})$. Then for $R \le n^{\ell}/\ell^2$, with probability at least $1-\exp(-\Omega_{\ell}(n))$, we have both $\sigma_R (U)$ and $\sigma_{R (R-1)/2} (M_{\Phi})$ to be $\ge \poly(1/n, \rho)$.
\end{lemma}
\begin{proof}
The desired inequality for the matrix $U$ was already shown in earlier sections. Let us thus consider $M_{\Phi}$.  We can write the $\{i,j\}$'th column as
\[ (M_{\Phi})_{i,j} = (\ta_i^{\ot \ell} \ot \ta_j^{\ot \ell}) - (\ta_i^{\ot (\ell-1)} \ot \ta_j \ot \ta_j^{\ot (\ell -1)} \ot \ta_i) + (\ta_j^{\ot \ell} \ot \ta_i^{\ot \ell}) - (\ta_j^{\ot (\ell-1)} \ot \ta_i \ot \ta_i^{\ot (\ell -1)} \ot \ta_j).  \]

We will show a stronger statement, namely that a matrix with four different columns (corresponding to each term above) for each pair $\{i,j\}$ has $\sigma_{\min} \ge \poly(1/n, \rho)$.  In this matrix, which we call $M_{\Phi}'$, we have two columns for every (ordered) pair $(i,j)$.  The first column is $\ta_i^{\ot \ell} \ot \ta_j^{\ot \ell}$ and the second is $\ta_i^{\ot (\ell-1)} \ot \ta_j \ot \ta_j^{\ot (\ell -1)} \ot \ta_i$.

For any of the columns, we thus have
\begin{align*}
& \Delta_2 = 1 \quad \text{ (same $i,j$, different terms)}\\
&\Delta_{\ell} = R-1 \quad \text{ (same $i$, different $j$)}\\
&\Delta_{\ell+1} = R-1 \quad \text{ (same $i$, different $j$, different terms)}\\
&\Delta_{2\ell-2} = 1 \quad \text{ ($i$ and $j$ swapped, different terms)}\\
&\Delta_{2\ell-1} = R-1 \quad \text{ (same $i$, different $j$, different terms)}\\
&\Delta_{2\ell} = R^2.
\end{align*}
\pnote{Fixed a couple things: $\Delta_0 = 0$ since we don't count the fixed column itself, so I removed this line. $\Delta_{\ell+2}$ changed to $\Delta_{\ell+1}$. $\Delta_{2\ell-2}$ is actually $1$, and added $\Delta_{2\ell - 1} = R-1$. None of these affect the conclusion at all.}

The rest of the $\Delta$ values are zero. Thus, we observe that we can apply Theorem~\ref{thm:indepmonomials2} (with $c=\Omega(1)$), as the dominant terms are the ones corresponding to $\Delta_2, \Delta_\ell, \Delta_{2\ell}$.  This completes the proof.
\end{proof}

%% file: appendix.tex
\section{Lemma~\ref{lem:deconditioning}}
\begin{lemma}\label{lem:deconditioning}
Let $X,Y$ be two independent real random variables, for all $a,b\in \mathbb{R}$ such that $\mathbb{P}[X+Y\le b]>0$, we have
\begin{align*}
    \mathbb{P}[X\le a]\le\mathbb{P}[X\le a| X+Y\le b]
\end{align*}
\end{lemma}
\begin{proof}
WLOG, assume $\mathbb{P}[X\leq a]>0,\mathbb{P}[X>a]>0$, then we have 
\begin{align*}
    \mathbb{P}[X+Y\leq b | X\leq a]&\geq\mathbb{P}[Y\leq b-a|X\leq a]=\mathbb{P}[Y\leq b-a]\\
    &=\mathbb{P}[Y\leq b-a|X> a]\geq \mathbb{P}[X+Y\leq b | X> a]
\end{align*}
Hence $\mathbb{P}[X+Y\leq b | X\leq a]\geq \mathbb{P}[X+Y\leq b]$
, then
\begin{align*}
    \mathbb{P}[X\leq a|X+Y\leq b]=\frac{\mathbb{P}[X\leq a]\mathbb{P}[X+Y\leq b|X\leq a]}{\mathbb{P}[X+Y\leq b]}\geq \mathbb{P}[X\leq a]
\end{align*}
\end{proof}

\section{Proof of Theorem~\ref{thm:newvectorCW}} \label{sec:newvectorCW}
The proof of Theorem~\ref{thm:newvectorCW} is almost identical to Theorem~\ref{thm:newdecoupling}.


\begin{proof}[Proof of Theorem~\ref{thm:newvectorCW}]
By Proposition~\ref{prop:decoupling}, it suffices to show that
\begin{align*}
    \mathbb{P}[\norm{\hat{g}(u+z_0,z_1,\cdots,z_{\ell-1})}_2<c(\ell)\epsilon\eta\cdot\frac{\rho^\ell}{n^\ell}]<\epsilon^{c'(\ell)\delta n}
\end{align*}
where $z_0\sim N(0,\rho^2(\ell+1)/(2n\ell))^n$ and $z_1,z_2,\cdots,z_{\ell-1}\sim N(0,\rho^2/(2n\ell))$, $c(\ell),c'(\ell)>0$ are constants depending only on $\ell$. Let $W$ be the span of the top $\delta n^\ell$ right singular vectors of $M$. Observe that
$$   \norm{\hat{g}(u+z_0,z_1,\cdots,z_{\ell-1})}_2=\norm{M (u+z_0)\otimes z_1\otimes\cdots\otimes z_{\ell-1}}_2 \ge \eta \norm{\Pi_W ~(u+z_0)\otimes z_1\otimes\cdots\otimes z_{\ell-1}}_2.$$

The theorem then follows by applying Lemma~\ref{lem:quantitativebound} with $x_1=u,x_2=x_3=\cdots=x_\ell=0$ and $p=\epsilon^{1/\ell}$.
\anote{4/2:Shortened the above proof.}

\end{proof}

\section{ Combinatorial Proof of Theorem~\ref{thm:newdecoupling} and Theorem~\ref{thm:newvectorCW} for $\ell=2$.} \label{sec:combinatorialproof}
We now give a self-contained combinatorial proof of Theorem~\ref{thm:newdecoupling} for $\ell=2$, that uses decoupling and Lemma~\ref{lem:quantitativebound}. Let $D'$ be the dimension of the subspace $W$, and let $M_1, M_2, \dots, M_{D'}$ be a basis for $W$. Let $z \sim N(0,\rho^2)^n$ and $z_1, z_2, \dots, z_r \sim N(0, \rho^2/r)^n$ be independent Gaussian random vectors for $r = O(\sqrt{n})$. Note that $x + z_1 \pm z_2 \pm z_3 \pm \dots \pm z_r$ are all identically distributed as $\tilde{x}$. 

Consider the following process for generating $\tilde{x}=x+z$. We first generate $z_1, z_2, \dots, z_r$ and random signs $\zeta=(\zeta_2, \zeta_3, \dots, \zeta_r) \in \set{\pm 1}^{r-1}$ all independently, and return $z=z_1+\sum_{i=2}^r \zeta_2 z_2$. It is easy to see that $z \sim N(0,\rho^2)$. We will now prove that at most one of the $2^{r-1}$ signed combinations $z_1 \pm z_2 \pm \dots \pm z_r$ has a non-negligible projection onto $W$.

Consider any fixed pair $\zeta, \zeta' \in \set{\pm 1}^{r-1}$, and let $u=z_1+\sum_{i=2}^r \zeta_i z_i$ and $u'=z_1+\sum_{i=2}^r \zeta'_i z_i$. We will use the basic decoupling Lemma~\ref{lem:deg-l} to show w.h.p. at least one of $\norm{\Pi_{W} u^{\otimes 2}}_2$ or $\norm{\Pi_{W} (u')^{\otimes 2}}_2$ is non-negligible. 
Using decoupling (with $\ell=2$) in Lemma~\ref{lem:deg-l} we have for each $j \in [D']$ 
\begin{align}
\Iprod{M_j, (x+u)^{\otimes 2}} - \Iprod{M_j, (x+u')^{\otimes 2} } & = 4\Iprod{M_j, (x+u+u') \otimes (u-u')} \nonumber \\
&= 4 \Iprod{M_j, (x+v_1) \otimes v_2}, \\
\text{ where } v_1&=z_1+\sum_{2\le i \le r: \zeta_i = \zeta'_i} \zeta_i z_i, ~ v_2 = \sum_{2\le i \le r: \zeta_i \ne \zeta'_i} \zeta_i z_i. \label{eq:bipartite}
 \end{align}

Also from Lemma~\ref{lem:quantitativebound}, we have that the above decoupled product $(x+v_1) \otimes v_2$ has a non-negligible projection onto $W$; hence with probability at least $1- \exp\big( - \Omega(\delta n) \big)$,
\begin{align*}
\norm{\Pi_{W} (x+v_1) \otimes v_2}_2^2 = \sum_{j=1}^{D'} \Iprod{M_j, (x+v_1) \otimes v_2}^2 &\ge  \frac{\Omega(\rho^4)}{r^2 n^{4}}\\
\text{ i.e., }~ \exists j^* \in [D'] \text{ s.t. } \abs{\Iprod{M_{j^*}, (x+v_1) \otimes v_2}} &\ge \frac{\Omega(\rho^2)}{r n^{3}}. 
 \end{align*}
 
Applying \eqref{eq:bipartite} with the above inequality for $j^*$,
\begin{align}
    \abs{\Iprod{M_{j^*}, (x+u)^{\otimes 2}} - \Iprod{M_{j^*}, (x+u')^{\otimes 2} }} &\ge  \frac{\Omega(\rho^2)}{r n^{3}}. \nonumber \\
    \text{Hence, } \norm{\Pi_W (x+u)^{\otimes 2}}_2 + \norm{\Pi_W (x+u')^{\otimes 2} }_2 &\ge \Omega\Big( \frac{\rho^2}{r n^{3}} \Big), \label{eq:decoup:pairwise}
\end{align}
with probability at least $1-\exp(-\Omega(\delta n))$

Since $r=c_1 \delta n$ (for a sufficiently small constant $c_1 >0$), we can apply \eqref{eq:decoup:pairwise} for each of the $2^{2r-1}$ pairs of $\zeta, \zeta' \in \set{\pm 1}^{r-1}$, and union bound over them to conclude that 
with probability at least $1-\exp(-\Omega(\delta n))$
$$\forall \zeta \ne \zeta' \in \set{\pm 1}^{r-1}, ~ \max \Big\{\big|\iprod{M_j, (x+z_1+\sum_{i=2}^r \zeta_i z_i )^{\otimes 2}}\big|,\big|\iprod{M_j, (x+z_1 + \sum_{i=2}^r \zeta'_i z_i)^{\otimes 2} }\big|\Big\} \ge \frac{\Omega(\rho^2)}{r n^{3}}.$$

Hence w.h.p. at most one of the $2^{r-1}$ signed combinations $x+z_1 \pm z_2 \pm \dots \pm z_r$ has a negligible projection onto $W$. Hence, with probability at least $1-2^{-r+1}$ i.e., with probability at least $1-2^{-\Omega(\delta n)}$, $\norm{\Pi_W \tilde{x}^{\otimes 2}}_2 \ge \Omega(\rho^2)/n^{4}$.  This establishes Theorem~\ref{thm:newdecoupling}. An identical proof also works for Theorem~\ref{thm:newvectorCW} when $\ell=2$.